\documentclass[11pt]{article}
\usepackage{amsfonts}
\usepackage{amssymb,amsmath,amsthm}
\usepackage{graphicx}
\usepackage[dvips, paper=letterpaper, top=1in, bottom=1in, left=1in, right=1in, nohead, includefoot, footskip=0.4in]{geometry}
\usepackage[authoryear,round]{natbib}
\usepackage{environ}
\usepackage{color}
\usepackage{epsfig}
\usepackage{caption}
\usepackage{verbatim}
\usepackage{subcaption}
\usepackage{float}
\usepackage{setspace}
\usepackage{enumitem}
\setlist[enumerate]{itemsep=0mm}
\usepackage{calc}
\usepackage[pdfpagemode=UseOutlines,hyperfootnotes=false, plainpages=false]{hyperref}
\usepackage{mathtools}
\usepackage{dsfont}
\usepackage[english]{babel}
\usepackage{amsthm}
\usepackage{bm}
\usepackage{breqn}
\usepackage{bigints}
\usepackage{tikz}
\usetikzlibrary{positioning}
\usepackage{nicefrac}
\usepackage{xurl}
\usepackage[linesnumbered,ruled,vlined]{algorithm2e}
\usepackage{algpseudocode}
\usepackage{subfiles}
\usepackage[capitalise]{cleveref}
\usepackage[flushleft]{threeparttable}
\usepackage{array,booktabs,makecell}
\usepackage{subcaption}
\usepackage{booktabs}
\usepackage{pgfplots}
\pgfplotsset{compat=1.18}

\usepackage[font=normalsize,labelfont=bf]{caption}
\usepackage[font=footnotesize,labelfont=bf]{subcaption}
\usepackage{sectsty}
\allsectionsfont{\normalfont\sffamily\bfseries}
\sffamily

\newtheoremstyle{theoremsansserif} 
    {\topsep}                    
    {\topsep}                    
    {\itshape}                   
    {}                           
    {\sffamily\bfseries }        
    {.}                          
    {.5em}                       
    {}  

\theoremstyle{theoremsansserif}

\newtheorem{lemma}{Lemma}
\newtheorem{corollary}{Corollary}
\newtheorem{proposition}{Proposition}

\newtheorem*{example*}{Example}
\newtheorem{assumption}{Assumption}
\newtheorem{theorem}{Theorem}

\theoremstyle{definition}
\newtheorem{exmp}{Example}

\newcommand{\R}{\mathbb{R}}
\newcommand{\Z}{\mathbb{Z}}

\newcommand{\D}{\mathcal{D}}
\newcommand{\M}{\mathcal{M}}
\newcommand{\A}{\mathcal{A}}
\newcommand{\E}{\mathbb{E}}
\renewcommand{\L}{\mathcal L}

\newcommand{\F}{\mathcal F}

\newcommand{\pp}{\mathcal{P}}
\newcommand{\B}{\mathcal B}

\newcommand{\floor}[1]{\lfloor #1 \rfloor}
\DeclareMathOperator{\opt}{OPT}
\DeclareMathOperator{\argmin}{argmin}
\DeclareMathOperator{\argmax}{argmax}

\DeclareMathOperator{\alg}{ALG}



\title{Coordinating Spot and Contract Supply in Freight Marketplaces}
\author{
\sf Philip Kaminsky \\ \sf Amazon\\ \texttt{philipka@amazon.com} 
\and \sf Rachitesh Kumar\footnote{Work performed while author was a Postdoctoral Scientist at Amazon's Middle Mile Marketplace Science team.} \\ \sf Carnegie Mellon University \\ \texttt{rachitesh@cmu.edu}
\and \sf Roger Lederman \\ \sf Amazon \\ \texttt{rllederm@amazon.com}
}
\date{}

\begin{document}

\setstretch{1.5}

\maketitle

\begin{abstract}
    The freight industry is undergoing a digital revolution, with an ever-growing volume of transactions being facilitated by digital marketplaces. A core capability of these marketplaces is the fulfillment of demand for truckload movements (loads) by procuring the services of carriers who execute them. Notably, these services are procured both through long-term contracts, where carriers commit capacity to execute loads (e.g., contracted fleet of drivers or lane-level commitments), and through short-term spot marketplaces, where carriers can agree to move individual loads for the offered price. This naturally couples two canonical problems of the transportation industry: contract assignment and spot pricing. In this work, we model and analyze the problem of coordinating long-term contract supply and short-term spot supply to minimize total procurement costs. We develop a Dual Frank Wolfe algorithm to compute shadow prices which allow the spot pricing policy to account for the committed contract capacity. We show that our algorithm achieves small relative regret against the optimal---but intractable---dynamic programming benchmark when the size of the market is large. Importantly, our Dual Frank Wolfe algorithm is computationally efficient, modular, and only requires oracle access to spot-pricing protocols, making it ideal for large-scale markets. Finally, we evaluate our algorithm on semi-synthetic data from a major Digital Freight Marketplace, and find that it yields significant savings ($\approx$10\%) compared to a popular status-quo method.
\end{abstract}



\newpage

\section{Introduction}\label{sec:intro}

The truckload freight industry is the circulatory system of the economy, generating close to a trillion dollars of revenue in the United States alone\footnote{https://www.trucking.org/economics-and-industry-data}. At its core, it is a market where shippers who wish to move freight trade with carriers (drivers)\footnote{We use \emph{carrier} and \emph{driver} interchangeably to denote the operating party that executes loads (including owner-operators), though legally they may be distinct.} who serve this demand by executing the freight movements. However, unlike other contemporary markets that have embraced digitization, like ride hailing and retail, the freight market still relies heavily on manual systems involving phone calls and human operators. The inefficiencies inherent in manual operations, combined with the industry’s trillion-dollar scale, have set the stage for a digital revolution: \textbf{Digital Freight Marketplaces (DFMs)} have emerged to improve efficiency and deliver better outcomes for shippers and carriers, and ultimately for drivers and customers.

Every day, these DFMs facilitate the execution of thousands of truckload movements (called \emph{loads}) via carriers across a nationwide road network. Operationally, they procure carrier capacity through two distinct channels. On the one hand, they sign long-term contracts spanning multiple months under which carriers commit to move loads for the DFM---for example, these can be lane-level commitments, a contracted fleet, or in-house drivers\footnote{We use \emph{contracts} broadly to encompass any form of capacity commitment with payment guarantees, whether sourced through carrier contracts, private fleet arrangements, or employment.}; these contracts are typically signed weeks or months before execution begins. On the other hand, they run short-term spot load boards, where individual loads are posted days to hours before execution and dynamically priced until a carrier accepts. The DFM must coordinate these two channels to cover all required loads at minimum procurement cost---a task often complicated by uncertainty in the volume and composition of loads. In particular, given a supply of committed contracts that have already been signed, how should a DFM assign loads to contracted capacity and price them on the spot marketplace in order to minimize total cost?

In practice, the contract and spot sides of the supply are often optimized in relative isolation. Contract teams focus on utilization and on-time performance for committed capacity. Spot pricing teams focus on load-by-load decisions, tuning dynamic pricing policies to clear the load board while controlling short-term spend. Coordination is typically handled via simple rules of thumb: earmarking certain lanes or volumes for contracts, sending “overflow’’ to spot, or reserving contract capacity for loads deemed critical. These heuristics ignore the selection power of the spot marketplace---carriers choose which loads to accept---and the flexibility embedded in contracts, especially contracted fleet capacity that can be used to route drivers across multiple loads. As a result, DFMs routinely find themselves in costly suboptimal regimes where they pay high prices on loads that attract little interest from carriers in the spot marketplace and should have been assigned to contracts, while using contracts for loads that would have attracted strong interest in the spot marketplace..

This paper takes a first step toward a principled algorithmic approach to this coordination problem. We propose and analyze a general model of a Digital Freight Marketplace that has to execute a set of loads on a fixed future date. The DFM has access to a collection of committed contracts, each of which can cover certain subsets of loads (capturing both lane-level and contracted fleets), and to a spot load board where loads are posted and priced dynamically over time. On the spot side, we model pricing for each load as a Markov decision process (MDP) with a flexible state space and a terminal cost capturing the penalty for failing to procure a carrier for it on the load board. This MDP framework encompasses static posted pricing, sequential posted pricing, and richer dynamic pricing heuristics currently used in the industry~\citep{uber-blog}. On the contract side, the DFM can allocate any subset of loads to contracts subject to the combinatorial structure induced by feasible routes and lane commitments. The benchmark is the optimal dynamic program that jointly chooses all spot prices over time and all contract assignments to minimize total expected cost.

Directly solving this benchmark is hopeless: the state space of the dynamic program tracks which subset of loads remains unprocured, which is exponentially large in the number of loads $L$. A natural response is to simplify the interaction between contract and spot supply. One popular heuristic used in practice is the \emph{Load Bifurcation Algorithm} (LBA): before any spot pricing takes place, the DFM chooses a subset of loads to assign to contracts and commits to procuring the remaining loads on the spot marketplace. This choice is made by computing, for each load, its optimal expected spot cost (including the terminal cost associated with a failure to procure on the spot marketplace), and then solving an integer optimization problem to pick which loads to assign to contracts. LBA is interpretable and easy to implement, making it the status quo at many digital freight marketplaces. However, we show that it can be highly sub-optimal even in simple symmetric examples: by committing ex ante which loads will go to contracts, LBA aims at minimizing \emph{expected} cost through maximal contract utilization, but ignores the \emph{distribution} of spot cost and the stochastic fluctuations in the preferences of the carriers in the spot marketplace. The key issue is priority. LBA reserves contract capacity up front and then sends whatever is left over to spot. In contrast, in a well-coordinated system, the spot marketplace should get the first right of refusal on all loads: carriers should be allowed to pick off loads they want to move, and the remaining loads should be soaked up by contracts and, if necessary, an alternate channel. We propose such an algorithm: Dual Frank Wolfe (Algorithm~\ref{alg:dfw}), which uses shadow prices for loads and contracts to coordinate the two types of supply and minimize cost.

\subsection{Main Contributions}

\paragraph{Model of coordinated spot and contract procurement.}
We introduce a model of digital freight procurement that jointly captures long-term contract capacity and short-term spot marketplaces. On the contract side, we allow contracts to cover combinatorial subsets of loads, encompassing everything from simple lane-level volume commitments on a fixed origin-destination pair to contracted fleet arrangements where a driver can execute multiple loads during a shift. On the spot side, the pricing problem is modeled as a general MDP. This framework is rich enough to include static posted pricing, sequential posted pricing, and reinforcement-learning-based pricing policies popular in practice. The central optimization captures the joint cost minimization of the total procurement cost, which includes both contract assignment and spot pricing.

\paragraph{Dual Frank Wolfe algorithm.}
To solve the global cost-minimization problem, we propose the Dual Frank Wolfe algorithm, which applies the Frank Wolfe (Conditional Gradient) algorithm in the dual space. It interprets dual variables as shadow prices and uses them to coordinate the two types of supply. In particular, it maintains shadow prices for loads and iteratively updates them to simulate a bargaining procedure between the contract supply and the spot marketplace, resulting in shadow prices which accurately internalize contract capacity and spot pricing curves. These shadow prices are then used as terminal costs for the spot pricing policy, thereby accounting for contract capacity. Our algorithm is modular---it only requires oracle access to the pricing algorithm used by the DFM, and is computationally efficient---it only requires one to solve a simple linear program per iteration.

\paragraph{Performance guarantees.}
Under mild regularity assumptions on the pricing oracle (Lipschitz continuity of the non-procurement probability in the terminal cost) and a structural substitutability condition on the feasibility of contract assignments, we show that the Dual Frank Wolfe policy attains an expected regret (additive gap) of $O(\sqrt{L} + \gamma \cdot \log_2(L))$ against the optimal dynamic-programming benchmark, where $\gamma$ captures the degree of substitutability in contract assignment. To do so, we leverage: (i) a fluid approximation that relates the expected contract-assignment cost for a random set of residual loads to a deterministic problem with demand vector $q$; (ii) an approximate convexity result showing that the minimum number of contracts needed to cover a given demand profile is uniformly close to a convex function; (iii) the vanishing Frank Wolfe gap in the dual space to bound primal error. Importantly, our algorithm does not require knowledge of the substitutability parameter $\gamma$ and the number of Frank Wolfe iterations needed to reach a prescribed dual gap does not depend on the number of loads $L$. Moreover, it is  amenable to a parallelized implementation, thereby making it ideally suited for large-scale markets. Finally, we also show that DFW ensures (near) full utilization of the committed contract capacity with high probability, thereby avoiding the operationally untenable pitfall of wasteful contract use.


\paragraph{Applications and Numerics.} We instantiate Dual Frank Wolfe (DFW) across three contract families and benchmark its performance against the status-quo Load Bifurcation Algorithm (LBA):
\begin{itemize}
\item \emph{Lane-level Contracts.} We prove that these contracts are perfectly substitutable, i.e., $\gamma = 0$, and DFW consistently outperforms LBA in line with our theoretical results. In particular, our numerical experiments with realistic spot price-conversion curves show that DFW yields relative savings as high as $50\%$ for such contracts.

\item \emph{Bipartite Round-Trip Contracts.}
Contracted fleet use in long-range freight networks is primarily governed by the return-to-domicile requirements which capture the fact that drivers need to return to their home domicile at the end of their shift. We encoded this round-trip-complementarity structure via a bipartite graph and numerically evaluated DFW against LBA on an instance created by calibrating the bipartite graph and spot price-conversion curves on real-world data from a major DFM. We found that DFW yields substantial relative savings, ranging  from 3.5--9.6\% depending on the cost of the alternate channel. This translates to millions of dollars in savings at modern DFMs which routinely transact billions of dollars in load movements every year. Furthermore, we numerically analyzed the impact of substitutability as captured by the degree of the bipartite graph: our experiments with random Erd\H{o}s--R\'{e}nyi bipartite graphs show that DFW yields over $10\%$ relative savings even when the average degree is only 5, with the marginal gains plateauing near 16\% by degree 10.

\item \emph{Regional Contracts}. These contracts model regional fleet operations where contracted drivers execute multiple loads within one area. The contract assignment is driven by two practical constraints: loads occupy the driver over load-dependent time intervals and drivers typically stay within a single subregion for the shift. We show that these combinatorial constraints do not materially preclude substitutability by establishing $\gamma = O(\sqrt{L})$, thereby proving that DFW achieves sub-linear $O(\sqrt{L})$-regret in the total load volume $L$.
\end{itemize}
Our numerical results demonstrate that DFW yields substantial savings on realistic instances, with the extent of the savings being governed by the amount of substitutability that can be accommodated by the contracts. Remarkably, in all our experiments on realistic bipartite instances, Dual Frank Wolfe converged in very few iterations, never once exceeding 50 iterations in total.

\subsection{Related Work}\label{sec:related-work}

Following the deregulation of the trucking industry after the Motor Carrier Act of 1980, the procurement methods for truckload movements have progressively become more sophisticated. This has given rise to a plethora of research problems that have been tackled by researchers from a wide range of domains. We do not aim to provide a comprehensive survey of this literature here. Instead, we refer the reader to the encyclopedic survey by \citet{acocella2023research}, which reviews the literature on truckload procurement in its entirety and positions it into various streams: carrier vs shipper perspectives, in-house vs outsourced capacity, long-term strategic vs short-term execution decisions etc. Of particular interest to us are the following observations made in \citet{acocella2023research}: \emph{``In our review of the literature, we do not find any peer-reviewed papers that take the shipper's perspective on whether to use in-house or for-hire TL transportation services at the execution stage."} and \emph{``However, the decision by the shipper to either execute a standing contract or to utilize a transactional, spot market option at the load tendering stage has received attention from only a few researchers."}. We make progress towards filling these gaps by devising and analyzing a concrete algorithm to coordinate execution decisions across different execution channels---namely committed capacity (contracted or in-house) and spot marketplaces.

Digital freight marketplaces (DFMs) typically procure capacity on the spot market through a mix of posted pricing and auctions. \citet{scott2018carrier} and \citet{scott2019concurrent} empirically analyze spot auctions for truckload movements, analyzing the impact of carrier features and the presence of concurrent contract relationships respectively. \citet{chen2021posted} study a freight marketplace that serves as an intermediary between shippers and carriers in a truckload transportation network. They compare posted-pricing versus hybrid mechanisms that combine posted pricing and auctions via a fluid approximation, showing that hybridization can strictly improve average profit relative to pure posted pricing. \citet{cao2025dynamic} study a dynamic model of the spot marketplace where shippers and carrier arrive online, modeling the resulting two-sided marketplace with expiring jobs as a long-run average-reward MDP. Our work is complementary in scope: we treat the spot-procurement mechanism as a black box, allowing it to be based on dynamic posted-pricing or any other MDP-based mechanism. Instead, we focus on the missing coordination layer between spot procurement and committed contract capacity.

The primary focus of the literature has been on the design of combinatorial auctions/mechanisms for establishing contracts \citep{caplice2006combinatorial, caplice2007electronic, sheffi2004combinatorial}. In particular, substantial attention has been given to the combinatorial winner determination problem \citep{buer2010solving, sandholm2002algorithm, sandholm2005cabob, lim2008transportation, guo2006carrier}. Primarily, these combinatorial auctions result in lane-level contracts, which fit neatly into our framework---yielding perfect substitutability and near-optimal performance for our Dual Frank-Wolfe algorithm. However, lane-based contracts contribute to the notorious \emph{deadhead miles} problem because they leads to empty movements~\citep{heilmann2020information}. The recent work of \citet{candogan2025taming} takes a mechanism design approach to procurement and analyzes a mechanism where a freight brokerage offers carriers cyclic tours and compensates drivers for on-duty time. In our model, this corresponds to a contracted fleet of drivers who are paid for their time. Importantly, as with the spot pricing literature, our work is complementary because we do not assume a specific contract type and design a modular coordination algorithm that works for any contract type so long as the assignment constraints are well-specified---as is the case with both lane-level and time-based contracts.

A number of works have looked at the problem from the carrier's perspective, e.g., what price to quote to a shipper who desires to move their loads on the carrier's network or how to bid in a combinatorial auction. Since we assume that the load set is fixed by the time of execution---as is typically the case in practice, we do not analyze the procedure that generated that load set. We refer the reader to the survey \citet{acocella2023research} and works on this topic by \cite{figliozzi2005impacts, figliozzi2007pricing, topaloglu2007incorporating}. Methodologically, our approach belongs to a broad class of decomposition and shadow-pricing methods in large-scale operational systems, where dual variables translate global capacity constraints into local decision signals. In particular, our algorithm makes use of the Frank Wolfe (or Conditional Gradient) Algorithm in the dual space. We refer the reader to the comprehensive survey by \citet{braun2022conditional} for an overview of the rich literature on this algorithm, including primal-dual versions for applications like Support Vector Machines~\citep{lacoste-julien13, osokin16}. More generally, other first-order methods have also been used to coordinate systems~\citep{maggiar2025consensus} and distribute optimization~\citep{neal2011distributed}.

\section{Model}\label{sec:model}

\paragraph{Notation.} We use $\Z_+$ (and $\R_+$) to denote the set of non-negative integers (and real numbers), while $\Z_{++}$ ( and $\R_{++}$) denotes the set of strictly positive integers (and real numbers). For a probability distribution $P$, $X\sim P$ denote s a  random variable distributed according to $P$. $[n] = \{1, \dots, n\}$ denotes the set of positive integers less than or equal to $n$.

Consider a Digital Freight Marketplace (DFM) which procures supply through both spot and contract markets. Let $\L$, with $L = |\L|$, denote the set of loads it must execute on a fixed date in the future, called the \emph{execution date}. A load here refers to a full-truckload movement with all the features necessary for its execution specified---like origin, destination, trailer type etc. We assume that the DFM must execute all loads $\L$ on the execution date. In practice, this is often the case because shippers---and as a consequence the DFM---are bound by delivery-date promises to customers which guarantee delivery by a specified date. In order to fulfill these promises, DFMs often use alternate channels that guarantee execution, e.g., human operators manually calling carriers. We use $a_\ell > 0$ to denote the cost associated with the use of the alternate channel to execute load $\ell$. Alternatively, in the absence of execution guarantees, $a_\ell$ can capture the cost imposed by the loss of good will caused by a missed load execution.

Contracts in the freight industry are long-term commitments between carriers and shippers, wherein a carrier commits to supply a certain capacity for load movements at a specified frequency (say daily), and the shipper commits to pay for that capacity. Given the long-term nature of these contracts (often months or years), these commitments are typically made far in advance of the execution date. Consequently, we assume that the DFM has already signed contracts with carriers and can use that capacity to execute loads on the execution date of interest. In particular, we assume that the DFM has $B$ contracts available to execute loads in $\L$. Each contract can be used to execute many different loads or collections of loads, and the DFM must decide the assignment of loads to contracts.  We use $\A \subseteq 2^\L$ to denote the feasible subsets of loads which can be executed via contracts, i.e., a subset $A \subseteq \L$ of loads can be assigned to a contract if and only if $A \in \A$. Naturally, these sets must be downward closed, i.e., if $\tilde A \subset A$, then $\tilde A \in \A$. This is because carriers can be routed with unused capacity (i.e., they can drive segments without any cargo), and thus any subset of loads from a feasible assignment can also be executed within the same contract. Define $u = \max_{A\in \A} |A|$, i.e., at most $u$ loads can be assigned to any single contract. In order to ensure that the coordination problem remains non-trivial, we assume that (i) every load $\ell \in \L$ can be assigned to a contract as part of some $A \in \A$~\footnote{If a load cannot be assigned to a contract then we would have no choice besides procuring it via the spot marketplace, and can thus safely ignore such a load for the purposes of the coordination algorithm.}; (ii) there exists $\nu \in (0,1/2)$ such that $B > \nu \cdot L$ and $u \cdot B < (1 - \nu) \cdot L$, i.e., there are sufficiently many contracts to execute a significant fraction of the loads, but not enough to execute all of them. 

Let $C: \{0,1\}^{|\L|} \to \R_+$ denote the minimal cost incurred in fulfilling the subset of loads $X \in \{0,1\}^{|\L|}$ (represented by an indicator vector) with the available contracts $B$, defined as
\begin{align*}
    C(X) \quad \coloneqq \quad \min\quad  &\sum_{i \in \L} a_i \cdot z_i\\
    \text{s.t.}\quad 
    & z_i\ +\ \sum_{A: i \in A} x_{A} \geq X_i &&\forall\ i \in \L\\
    & \sum_{A\in \A} x_{A} \leq B\\
    &z_i, x_{A} \in \{0,1\}
\end{align*}
Here, $x_{A}$ denotes the decision variable which is 1 if and only if the collection of loads $A \in A$ is assigned to a contract, and $z_i$ denotes the decision to use the alternate channel for load $i$.

Broadly, there are two major families of contracts currently in use at Digital Freight Marketplaces\footnote{\url{https://relay.amazon.com/blog/trucking-work-amazon-relay}}:
\begin{itemize}
    \item \textbf{Lane-level Contracts.} Lanes refer to a pair of origin and destination regions, defined at some geographical granularity like zip code, city, domicile etc. Each lane-level contract is for loads on a particular lane. In such a contract, the carrier agrees to execute a specified number of loads---called the volume---every week. For example, the carrier might agree to execute 100 loads every week which originate in the New York Metropolitan area and end up in the Chicago Metropolitan area. To model such contracts in our framework, we let $\L$ be the set of all loads on a lane and let $B$ denote the total volume of all contracts on it. In this case, the compatibility set is the collection of all singleton sets $\A = \{ \{\ell\} \mid \ell \in \L\}$.

\item \textbf{Contracted Fleet.} In these agreements, drivers in the contracted fleet---whether contracted carriers, private fleet, or in-house employees---make their capacity available during recurring shifts (duty periods) each week. During a driver's shift, the DFM may assign multiple loads for the driver to execute in succession. Since drivers typically prefer to return home at the end of a shift, the primary routing constraint is that each driver's route must originate in (and return to) the contracted domicile area by the end of the shift. Additionally, standard operational constraints apply, such as maximum continuous driving limits and minimum rest requirements. Our framework captures contracted-fleet agreements by setting $\A$ to be the set of feasible routes---each comprising several loads---that satisfy the domicile and duty-period constraints specified by the contract. Many DFMs also offer dedicated/tour-style arrangements in which drivers provide recurring capacity and are compensated for their time, while the DFM assembles and assigns multi-load tours (e.g., \citealt{candogan2025taming}). In our model, these are the same object: a capacity commitment defined by a feasible set of routes rather than a fixed lane.

\end{itemize}

On the other hand, DFMs also procure supply via the spot marketplace, where carriers and shippers transact loads with a short lead time to execution. These markets operate via \emph{load boards} where shippers post the loads they would like to execute and the price they are willing to offer for them. These prices are updated over time till a carrier chooses to fulfil the load at the posted price. We model the dynamics of the spot marketplace using an Markov Decision Process (MDP), which is in line with practice~\citep{uber-blog} and prior work~\citep{cao2025dynamic, chen2021posted}. In particular, we associate an MDP $\M_\ell = (T_\ell, S_\ell, O_\ell, \{B_{\ell,t}\}, w_\ell)$ with each load $\ell \in \L$, where 
\begin{itemize}
    \item $T_\ell$ is the number of time periods for which load $\ell$ remains on the load board, i.e.,  the number of times the price can be updated between the time it is posted and when it must be executed.
    \item $S_\ell$ is the state space. This is a design choice made by the DFM according to its needs. It can range from the state space that simply captures whether or not the load has been accepted by a carrier, to the much richer space consisting of all historical information.
    \item $O_\ell$ denotes the set of possible prices that can be offered for load $\ell$ in any period.
    \item $B_{\ell,t}(p; s)$ denotes the probability with which load $\ell$ is accepted in period $t$ when the state is $s \in S_\ell$ and the price offered is $p \in O_\ell$. If the load is accepted, the cost is equal to the price $p$.
    \item $w_\ell$ denotes the terminal cost, which is the cost incurred if  load $\ell$ is not accepted by the end of period $T_\ell$.
\end{itemize}
We assume that the possible set of prices $O_\ell$ includes the price $0$, and the corresponding probability of acceptance $B_{\ell, t}(0;s)$ is always 0. Therefore, pricing is sufficient to prevent loads from being accepted on the load board, and we assume without loss that all loads $\L$ are posted on the load board. We use $\pp_\ell$ to denote the set of all implementable pricing policies for load $\ell$. We allow the pricing policies to be random, i.e., they can set random prices in each period. The set $\pp_\ell$ is determined by the structure of the MDP $\M_\ell$ and the algorithm used to optimize the pricing policy. Let $ S_\ell(P)$ be the expected cost incurred by the pricing policy $P \in \pp_\ell$ when $w_\ell = 0$, and $Q_\ell(P)$ be the corresponding probability of \emph{not procuring} a carrier for load $\ell$ by the end of period $T_\ell$. We assume that the Bernoulli distributions $\{Q_\ell(\cdot)\}_{\ell \in \L}$ corresponding to spot procurement are independent across loads; this is the case in practice when the market is large and carriers have independent heterogeneous preferences. For example, consider the following ways of modeling carrier preferences, which are popular in both theory and practice
\begin{itemize}
    \item \textbf{Static Pricing:} Due to operational reasons, DFMs often prefer to post a fixed take-it or leave-it offer price without making any changes to it. This prevents strategic behavior from carriers, like cancellation and re-booking of previously booked loads to obtain better prices. The static pricing problem can be formulated as a simple MDP with a single period, i.e., $T=1$, with acceptance probability $B_{\ell, 1}(p;s) = F_\ell(p)$, where $F_\ell$ is the distribution of the minimum price that carriers are willing to accept for load $\ell \in \L$.
    \item \textbf{Sequential Posted Pricing:} In richer models, the DFM can update prices over time as carriers arrive, leading to a sequential posted pricing problem with independent demand. In each period $t \in [T_\ell]$, the DFM posts a price $p_t$ for load $\ell$, and a carrier with an independently drawn private minimum-willingness-to-accept $v_t \sim F_{\ell, t}$ accepts the load if and only if $v_t \geq p_t$. This fits naturally into our MDP framework: the state $S_\ell$ can encode whether or not the load has been accepted (and possibly additional features such as market covariates), and the acceptance probability $B_{\ell,t}(p;s)$ is induced by the distribution of carrier values and the arrival process, as in the dynamic pricing models of \citet{cao2025dynamic}. Furthermore, we can restrict the class of implementable pricing policies $\pp_\ell$ to reflect real-world constraints like minimum prices, frequency of price change, monotonicity etc.
\end{itemize}

\textbf{Benchmark.} The DFM would like to make optimal use of its contracts and load board to minimize total procurement cost for $\L$. As contract assignments can be made close to the execution date, the DFM first posts all loads $\L$ on the load board, pricing them with policies $\{P_\ell\}$, and then assigns the unaccepted ones to contracts. In other words, it aims to solve the following optimization problem
\begin{align}\label{eq:benchmark}
    \opt\ \coloneqq \quad \min_{P_\ell \in \pp_\ell} \quad \sum_{\ell \in \L}\ S_\ell(P_\ell)\ +\ \E_{X_\ell \sim Q_\ell(P_\ell)}[C(X)]\,.
\end{align}

Our goal is to develop an algorithm which (approximately) solves \eqref{eq:benchmark}. Note that it is possible to solve \eqref{eq:benchmark} via a global dynamic program which both prices loads and assigns them to contracts. However, such a dynamic program is computationally intractable due to the curse of dimensionality: the contract assignment problem couples the pricing problems of all loads and the state space of the dynamic program includes the subset of unaccepted loads in each period, which is exponentially large and computationally intractable. This presents the need for computationally-efficient algorithms which can coordinate contract assignment and load board pricing. Moreover, for the algorithm to be practical, it must be modular, i.e., it should work for wide variety of pricing MDPs and policies. Digital Freight Marketplaces are constantly innovating on the pricing strategies they use. In this rapidly changing environment, modularity ensures that the algorithm does not overfit to a particular pricing model, allowing it to endure changes and remain effective at the core task of coordinating spot and contract supply.

\section{Warm-up: Load Bifurcation Algorithm}\label{sec:naive-alg}

Before describing our algorithm, it is useful to analyze a popular heuristic which is often the status quo in real-world digital freight markets because it appears to naturally address the coordination problem at first glance. This heuristic bifurcates the set of loads: it picks a subset of loads $S \subset \L$ to assign to contracts and procures the remainder of the supply on the spot marketplace; we will refer to it as the \emph{Load Bifurcation Algorithm} (LBA). The choice of $S$ is made to minimize the expected cost of procuring the remainder on the spot marketplace. Concretely, it first computes $\alpha_\ell \coloneqq \min_{P_\ell \in \pp_\ell}\ S_\ell(P_\ell) + a_\ell \cdot Q_\ell(P_\ell)$, which is the minimal expected cost associated with procuring load $\ell$ on the spot marketplace when the terminal cost is given by $w_\ell = a_\ell$, i.e., it is the cost associated with procuring a carrier on the spot marketplace when contracts are not available and the alternate channel must be used at a cost of $a_\ell$ if the carrier is not procured. LBA then selects the subset $S = \{\ell \in \L \mid z_\ell = 0\}$ to minimize expected cost by solving the following integer optimization problem:
\begin{align*}
    \min\quad  &\sum_{i \in \L}\ \alpha_i \cdot z_i\\
    \text{s.t.}\quad 
    & z_i\ +\ \sum_{A: i \in A} x_{A} \geq 1 &&\forall\ i \in \L\\
    & \sum_{A\in \A} x_{A} \leq B\\
    &z_i, x_{A} \in \{0,1\}
\end{align*}

LBA enjoys the benefits of modularity and interpretability: it simply assigns a subset of loads to contracts with the goal of minimizing the \emph{expected cost} of procuring the remaining loads via the spot marketplace and alternate channels. However, this approach can be highly sub-optimal, as the following example demonstrates.

\begin{exmp}\label{example:naive-bad}
    Suppose the DFM needs to execute $L = 1000$ loads. It has $B = 700$ contracts available, each of which can execute any one of the loads, i.e., $\A = \{\{\ell\} \mid \ell \in \L \}$. The DFM uses simple static pricing to procure carriers for loads on the spot marketplace and it is known that minimum price at which carriers would accept any load is uniformly distributed on $[100, 200]$ for every load, i.e., if the DFM posts a price of $p_\ell \in [100, 200]$, the probability of it being accepted by a carrier is $(p_\ell-100)/100$. Moreover, suppose the cost of the alternate channel is 300. In this case, we have $\alpha_\ell = \min_{p_\ell \in [100,200]}\ (p_\ell-100)\cdot p_\ell/100 + 300 \cdot (200 - p_\ell)/100 = 200$. 
    
    Since the loads are completely symmetric, LBA picks any $700$ and assigns them to contracts and procures the remainder via the spot marketplace at the cost of $200\cdot 300 = 60,000$. On the other hand, consider the following simple alternate algorithm: post a price of $p_\ell = 130$ for every load. Thus, each load on the spot marketplace is accepted with probability $0.3$. In other words, $300$ loads are accepted on average in the spot marketplace. The remaining 700 unaccepted loads are then assigned to a contract. Therefore, the total cost is $\approx 300\cdot 130 = 39,000$, which is much better than LBA's cost of $60,000$. This number is approximate because we have ignored the random noise: sometimes there will be more than 700 loads which are unaccepted and we will incur a cost of 300 for each of them. Rigorously accounting for this yields an exact expected cost that is bounded above by $41,000$, which is still substantially better than $60,000$. $\qed$
\end{exmp}

Before proceeding to our results, it is worth delving deeper into the phenomenon highlighted by Example~\ref{example:naive-bad}.
Note that LBA prioritizes contract assignment and aims to fully-utilize the contract capacity, i.e., it commits to selecting 700 loads for contracts. At first glance, this appears to be a worthy goal. After all, the DFM has paid for the contracted capacity and maximizing its utilization should result in lower costs. However, 100\% contract utilization comes at a cost: setting aside $B$ loads for contracts means that LBA is only able to account for the \emph{average} cost of procuring loads on the spot marketplace and ignores its \emph{distribution}. In particular, it only accounts for the preferences of the carriers on the spot marketplace at an aggregate level, thereby ignoring variance and idiosyncrasies, e.g., a carrier might happen to prefer a load because it fills an empty backhaul and is willing to execute it for a low price. 

The alternate approach on the other hand posts \emph{all loads} on the spot marketplace without bifurcating, and lets the carriers choose the loads they like---as defined by their willingness to execute at a lower price. The remainder are then assigned to contracts. This priority order naturally reflects the pricing ability possessed by the two types of supply: contracted carriers are committed to prices, whereas the spot carriers have the power to choose their price because they bear the risk of being left without any loads to execute. Therefore, LBA ends up paying a premium for ignoring this fact. In fact, this premium is even larger than the typical information rent incurred when pricing under limited information about the private value of the customer. In particular, even if we endow LBA with the ability to perfectly price loads it sends to the spot marketplace, i.e., every load it sends to the spot marketplace has a cost uniformly distributed between $[100,200]$, it still incurs a higher cost than the alternate algorithm which lacks this ability and is handicapped to simply post prices. To see this, note that the cost of LBA in this perfect pricing scenario is $300 \cdot 150 = 45,000$ because it procures 300 carriers on the spot marketplace at an average cost of $150$ (mean of uniform on $[100, 200]$), which remains higher than $41,000$.

Example~\ref{example:naive-bad} illustrates that coordination is not just a cosmetic detail: using contracts intelligently can create large savings even in settings with completely symmetric loads and simple static pricing. The key lesson is that the DFM should not irrevocably bifurcate loads into ``contract" and ``spot" sets ex ante. Instead, it should treat contracts as a flexible backstop that is deployed after the carriers in the spot marketplace have been given the opportunity to exercise their preferences over loads. In the remainder of the paper, we develop a coordination algorithm that formalizes this idea. Our algorithm generalizes the alternate strategy in Example~\ref{example:naive-bad} to arbitrary pricing MDPs and contract structures, while preserving the modularity of LBA: it only interacts with the pricing and contract-assignment mechanisms through black-box oracles, but it uses them jointly rather than in isolation.
\section{Dual Frank Wolfe}\label{sec:alg}

In this section, we develop an efficient algorithm for coordinating contract assignment and load-board pricing that is provably close to optimal. The key idea is to work in the dual space and interpret the dual variables as shadow prices: each load $\ell$ has an associated shadow price $\lambda_\ell$ and the contract capacity has shadow price $\mu$. Once these are fixed, the spot-pricing problem decomposes completely across loads: for each load $\ell$, we invoke a pricing oracle with terminal cost $w_\ell = \lambda_\ell$ which returns the optimal pricing policy and its non-procurement probability $q_\ell(\lambda_\ell)$, yielding a residual demand profile for contracts. On the contract side, we respond to this residual demand by solving a simple dual LP to find new shadow prices $(\hat \lambda, \hat \mu)$ consistent with the available capacity. The Frank Wolfe step then acts as a bargaining procedure between the spot and contract sides---by gradually adjusting the shadow prices $(\lambda,\mu)$ towards $(\hat \lambda, \hat \mu)$---until they approximately agree. The resulting procedure is easy to implement and highly modular: it sits on top of existing pricing and contract-assignment systems, allowing them to operate with minimal interference.

\subsection{Preliminaries and Assumptions}

Before describing our algorithm, we establish the necessary prerequisites. First, to ensure modularity, we posit the existence of a pricing oracle which can find the optimal pricing policy within the set of implementable pricing policies. 

\begin{assumption}[Pricing Oracle]\label{assum:pricing-oracle}
    For every load $\ell \in \L$ and terminal cost $w_\ell$, an optimal implementable pricing policy $P^*_\ell$ can be computed efficiently:
    \begin{align*}
        P_\ell^*(w_\ell)\ \in \argmin_{P \in \pp_\ell}\ S_\ell(P)\ +\ Q_\ell(P) \cdot w_\ell\,. 
    \end{align*}
    Moreover, the optimal implementable policy satisfies $Q_\ell(P_\ell^*(0)) = 1$ whenever the terminal cost $w_\ell = 0$.
\end{assumption}

Assumption~\ref{assum:pricing-oracle} formalizes the idea that the DFM can ``plug in'' its existing pricing engine as a black box. For any load $\ell$ and any terminal cost $w_\ell$, the oracle returns an implementable pricing policy that is optimal within the class $\pp_\ell$ when the DFM correctly internalizes the penalty $w_\ell$ for failing to procure a carrier for the load. The objective $S_\ell(P) + Q_\ell(P)\,w_\ell$ is exactly the total expected procurement cost under policy $P$, combining the payments made to carriers on the spot marketplace and the expected terminal cost. In practice, many DFMs already maintain sophisticated pricing systems for their load boards~\citep{uber-blog}. In simple cases, such as static pricing or small-state MDPs, the oracle can be implemented by solving the dynamic program exactly. In more complex settings, $\pp_\ell$ may be a parametric policy class (for example, prices of the form $p_t = f_\theta(s_t)$), and the oracle may use approximate dynamic programming or reinforcement learning to optimize $\theta$. Typical implementations include approximate value iteration with function approximation, policy-gradient or actor--critic methods over a fixed policy class, or deep RL heuristics to search for a good pricing policy. Assumption~\ref{assum:pricing-oracle} does not require global optimality over all measurable policies; it only requires that, for each $\ell$ and $w_\ell$, the pricing engine can efficiently return a policy that is optimal within the implementable policy class $\pp_\ell$. In particular, Assumption~\ref{assum:pricing-oracle} is compatible with approximate DP/RL implementations; as long as the pricing system optimizes over the same policy class $\pp_\ell$, our analysis goes through.


In order to simplify notation, we set $s_\ell(w_\ell) \coloneqq S_\ell(P^*_\ell(w_\ell))$ and $q_\ell(w_\ell) \coloneqq Q_\ell(P^*_\ell(w_\ell))$. The former represents the expected cost incurred by the policy $P_\ell^*(w_\ell)$ by the end of period $T$---disregarding the terminal cost, and the latter represents the probability of $P_\ell^*(w_\ell)$ not procuring the load by the end of period $T$. Furthermore, we set $r_\ell(w_\ell) \coloneqq s_\ell(w_\ell) + q_\ell(w_\ell) \cdot w_\ell$ to be the cumulative expected cost incurred by the optimal policy $P_\ell^*$ when the terminal cost is $w_\ell$.

\begin{assumption}[Lipschitz Conversion Rate]\label{assum:lipschitz}
    For every load $\ell \in \L$, there exists $\beta_\ell > 0$ such that $q_\ell(\cdot)$ is $\beta_\ell$-Lipschitz, i.e., $|q_\ell(w) - q_\ell(\tilde w)| \leq \beta_\ell \cdot |w - \tilde w|$ for all $w, \tilde w \in \R_+$.
\end{assumption}

Assumption~\ref{assum:lipschitz} requires that the non-procurement probability $q_\ell(w)$ of the pricing oracle is not overly sensitive to the terminal cost $w$. Intuitively, $w_\ell$ is a penalty for failing to procure a carrier for the load. Increasing $w_\ell$ encourages the pricing oracle to try harder to procure a carrier for the load on the spot marketplace, and so $q_\ell(w)$ should decrease with $w$. The Lipschitz condition says that this response is not arbitrarily steep; a unit change in the penalty cannot change the non-procurement probability by more than $\beta_\ell$. This is a common regularity condition in pricing and is satisfied by many systems in practice. For instance, in static pricing with a regular carrier-value distribution, the optimal price changes smoothly with the penalty, and the resulting non-procurement probability is Lipschitz in $w$. In the following lemma, we leverage Assumption~\ref{assum:lipschitz} to show that the expected cost of the pricing oracle $r_\ell(w)$ is $\beta_\ell$-Lipschitz. This follows from the Envelope Theorem: the derivative of $r_\ell(w)$ with respect to $w$ is exactly the non-procurement probability under the optimal pricing policy, that is, $r_\ell'(w) = q_\ell(w)$.

\begin{lemma}\label{lemma:smoothness}
    If \Cref{assum:lipschitz} holds, then $r_\ell(\cdot)$ is $\beta_\ell$-smooth.
\end{lemma}

Next, we describe the structural assumptions needed on the contract assignment problem. Define $\B(X) \coloneqq \min \{\sum_{A \in \A} x_A \mid \sum_{A: i \in A} x_{A} \geq X_i\ \forall\ i \in \L; x_A \in \Z_+\ \forall\ A \in \A\}$ to be the minimum number of contracts needed to completely cover the set of loads $X$. Note that this definition naturally applies to not just sets of loads $X \in \{0,1\}^L$, but also to multisets $X \in \Z_+^L$. In particular, it allows for copies of each load $\ell \in \L$ and treats copies as fungible for the purposes of contract assignment. Next, define $\B(d) = \E_{Z \sim d - \floor{d}}[\B(\floor{d} + Z)]$ to be the multi-linear extension of this set function for $X \in \R_+^{|\L|}$.  We abuse notation and use $X \sim d$ to denote the random vector $X = \floor{d} + Z$ where $Z \sim d - \floor{d}$.

\begin{assumption}[Substitutability in Contract Assignment]\label{assum:contract-subs}
    There exists $\gamma \in [0, B/4)$ such that
    \begin{align*}
        \B\left( \frac{X}{2} \right) \leq \frac{\B(X)}{2} + \gamma \qquad \forall\ X \in \R_+^L\,.
    \end{align*}
\end{assumption}

Assumption~\ref{assum:contract-subs} encodes a structural notion of substitutability in contract assignment. The inequality $\B(X/2) \le \B(X)/2 + \gamma$ says that if we scale a demand profile $X$ down by a factor of two, then we can cover this smaller set of loads using roughly half as many contracts, up to an additive slack $\gamma$. In operational terms, this means that contracts are not highly specialized to particular configurations of loads: capacity can be reassigned and repacked so that a large fraction of the original contracts remain useful when the underlying load pattern is sub-sampled uniformly at random. The parameter $\gamma$ measures the inherent indivisibilities and combinatorial frictions in the contract system; small $\gamma$ corresponds to a high degree of substitutability, where contracts can be flexibly shifted across different subsets of loads with only a modest loss in efficiency. Importantly, our Dual Frank Wolfe algorithm does not require knowledge of $\gamma$ and only the performance guarantee depends on it.

This assumption is consistent with many contract structures encountered in practice. In lane-level contracts, where each contract covers a single load on a lane, we have $\B(X) = \sum_\ell X_\ell$, so the inequality holds with $\gamma = 0$ (see \Cref{prop:lane-level-gamma}). In contracted fleets, each contract corresponds to a driver on duty who can execute several loads during their shift. Provided that (i) there are many overlapping feasible routes in the network, (ii) the maximum number of loads per contract $u$ is small relative to the total number of loads, and (iii) the load pattern is not adversarially aligned to a single route, the DFM can typically re-pack loads across contracts so that halving the demand approximately halves the number of contracts used, up to a minor adjustment. For example, these conditions are satisfied in regional transportation networks (see \Cref{prop:gamma-interval-graph}). Intuitively, Assumption~\ref{assum:contract-subs} rules out highly structured worst cases and focuses on the practically relevant regime in which contract capacity is used to cover a large and heterogeneous pool of loads, leading to many near-optimal ways to assign contracts. We provide specific examples of contracts and their $\gamma$ values in Section~\ref{sec:numerics}. Moreover, this assumption is also empirically supported by data from the digital freight market we study (see Section~\ref{sec:numerics}).

\begin{lemma}\label{lemma:approx-convexity}
    There exists a convex function $g: \R_+^L \to \R_+$ such that  $$g(d) \leq\ \B(d)\ \leq\ g(d) + 2 \cdot \gamma \cdot \log_2(4L)$$ for all $d \in \R_+^L$.
\end{lemma}

Lemma~\ref{lemma:approx-convexity} states that $\B(d)$ is uniformly close to a convex function $g(d)$: for every demand profile $d$, the value $\B(d)$ lies between $g(d)$ and $g(d) + 2\gamma \log_2(4L)$. In other words, $\B$ is approximately convex up to an additive error that is controlled by the substitutability parameter $\gamma$. Lemma~\ref{lemma:approx-convexity} formalizes the fact that, in realistic freight networks with substantial flexibility in contract assignment, the effective number of contracts needed is very close to a convex function of the underlying demand profile. Its proof combines Assumption~\ref{assum:contract-subs} with an intricate coupling argument and results on the stability of functional equations in several variables~\citep{hyers2012stability}.


\subsection{The Fluid Problem and Its Dual}

Our algorithm operates in the dual space. We now define the fluid relaxation of the contract assignment problem and describe its dual problem. For a fixed $q \in [0,1]^L$, define the fluid relaxation of the planning problem as
\begin{align*}
    \F(q) \quad \coloneqq \quad \min_{p,d} \quad &a^\top p\\
    \text{s.t.} \quad &p + d \geq q\\
    &\B(d) \leq B\\
    &d,\ p\ \in \R_+^L
\end{align*}

Given a vector $q \in [0,1]^L$ of non-procurement probabilities, $\F(q)$ treats $q$ as a fractional demand profile and chooses how much of each load to cover by contracts ($d$) and how much to send to the alternate channel ($p$). It replaces random quantities with their expectations: the objective $a^\top p$ is the expected alternative cost, and $\B(d)$ is the total expected contract usage. Lemma~\ref{lemma:fluid-error} shows that this fluid relaxation approximates the underlying integer problem up to the $O(\sqrt{L})$ term which accounts for the error caused by ignoring the randomness.

\begin{lemma}\label{lemma:fluid-error}
    For every $q \in [0,1]^L$, we have $\E_{Y \sim q}[C(Y)]\ \leq \F(q) + A_\max \cdot \sqrt{L}/2\,.$
\end{lemma}

The Lagrangian dual function of $\F(q)$ is given by
\begin{align}\label{eq:dual_def}
    \D(\lambda, \mu \mid q)\ \coloneqq \quad  \min_{p,d \in \R_+^L}\ (a - \lambda)^\top p + \left( \mu \cdot \B(d) - \lambda^\top d \right) + \lambda^\top q - \mu \cdot B\,. 
\end{align}

\begin{lemma}\label{lemma:dual-problem}
    For all $q \in [0,1]^L$, the dual optimization problem of $\F(q)$ is given by
    \begin{align*}
        \max_{\lambda, \mu \geq 0}\ \D(\lambda, \mu \mid q) \quad=\quad \max_{\lambda, \mu} \quad &\lambda^\top q - \mu \cdot B\\
        \text{s.t.} \quad &\sum_{\ell \in A} \lambda_\ell \leq \mu &&\forall\ A \in \A\\
        &\lambda_\ell \leq a_\ell &&\forall\ \ell \in \L\\
        &\lambda \in \R_+^L,\ \mu \geq 0\,.
    \end{align*}
\end{lemma}

The dual function $\D(\lambda,\mu \mid q)$ captures Lagrange multipliers on the fluid constraints: $\lambda_\ell$ is a shadow price for not covering load $\ell$ (instead of paying $a_\ell$ on the alternate channel), and $\mu$ is a price for using an additional unit of contract capacity in expectation. Lemma~\ref{lemma:dual-problem} shows that the dual of $\F(q)$ has a clean and interpretable form. The constraint $\lambda_\ell \le a_\ell$ prevents the shadow price of a load from exceeding its fallback alternative cost, while the constraints $\sum_{\ell \in A} \lambda_\ell \le \mu$ for all $A \in \A$ say that the total shadow value of any feasible subset of loads cannot exceed the price $\mu$ of one contract.

Till now, we have focused only on the contract assignment problem. We now broaden our purview and also consider the spot pricing problem. In particular, replacing the contract assignment problem in \eqref{eq:benchmark} with the dual optimization problem yields
\begin{align*}
    \min_{P_\ell \in \pp_\ell}\ \sum_{\ell \in \L} S_\ell(P_\ell)\ +\ \max_{\lambda, \mu \geq 0}\ \D(\lambda, \mu \mid Q(P_\ell))\,.
\end{align*}

Exchanging the $\min$ and the $\max$, i.e., considering the dual problem, yields
\begin{align*}
    \max_{\lambda ,\mu \geq 0}\ \quad &\left\{\sum_{\ell \in \L} \min_{P_\ell \in \pp_\ell} \big( S_\ell(P_\ell) + \lambda_\ell \cdot Q_\ell(P_\ell)\big)\right\} - \mu\cdot B
\\
    \text{s.t.} \quad &\sum_{\ell \in A} \lambda_\ell \leq \mu &&\forall\ A \in \A\\
        &\lambda_\ell \leq a_\ell &&\forall\ \ell \in \L
\end{align*}

Finally, using the definition of $r_\ell(w_\ell) = \min_{P_\ell \in \pp_\ell} S_\ell(P_\ell) + Q(P_\ell) \cdot w_\ell$ yields
\begin{align}\label{eq:combined-dual-problem}
    \max_{\lambda, \mu \geq 0} \quad & \sum_{\ell\in \L} r_\ell(\lambda_\ell) - \mu \cdot B\\
        \text{s.t.} \quad &\sum_{\ell \in A} \lambda_\ell \leq \mu &&\forall\ A \in \A \notag\\
        &\lambda_\ell \leq a_\ell &&\forall\ \ell \in \L \notag\,.
\end{align}

Note that plugging the dual of the fluid contract problem into the benchmark \eqref{eq:benchmark} lets us decouple the spot-pricing decisions across loads, as evident in \eqref{eq:combined-dual-problem}. For fixed dual variables $(\lambda,\mu)$, each load $\ell$ faces a modified pricing problem in which failing to procure a carrier for the load incurs a penalty $\lambda_\ell$. The inner minimization $\min_{P_\ell \in \pp_\ell} \big( S_\ell(P_\ell) + \lambda_\ell Q_\ell(P_\ell)\big)$ is exactly $r_\ell(\lambda_\ell)$, the optimal expected cost for load $\ell$ when the terminal cost is $\lambda_\ell$. The combined dual problem \eqref{eq:combined-dual-problem} therefore chooses shadow prices $(\lambda,\mu)$ to balance two forces: on the one hand, increasing $\lambda_\ell$ encourages the pricing oracle to procure a carrier for load $\ell$ on the spot marketplace; on the other hand, the constraints $\sum_{\ell \in A} \lambda_\ell \le \mu$ and $\lambda_\ell \le a_\ell$ ensure that these penalties remain consistent with what can actually be achieved using the finite contract capacity $B$ and the alternate channel. Our Dual Frank Wolfe algorithm will operate directly on this dual problem, using the pricing oracle as a subroutine.



\subsection{Algorithm}

\begin{algorithm}[t] 
   \caption{Dual Frank Wolfe}
   \label{alg:dfw}
    \begin{algorithmic}\vspace{0.08cm}
            \item[\textbf{Initialize:}] $\lambda_\ell = 0$ for all $\ell \in \L$; $\mu = 0$; step size $\eta_t = 2/(t+2)$; tolerance $\epsilon > 0$; $g = \infty$
            \item[\textbf{While}] $g > \epsilon$:
            \begin{itemize}
                \item \textbf{Query oracle} to get $\{q_\ell(\lambda_\ell)\}_\ell$ for all loads $\ell \in \L$.
                \item \textbf{Solve Dual LP.}
                \begin{align}\label{LP}
                    (\hat \lambda, \hat \mu) \quad \in \quad \argmax\quad &\tilde\lambda^\top q(\lambda)\ -\ \tilde \mu \cdot B \notag\\
                    \text{s.t.}\quad & \sum_{\ell \in A} \tilde \lambda_\ell \leq \tilde \mu &&\forall\ A \in \A \notag\\
                    &\tilde \lambda \leq a \notag\\
                    &\tilde \lambda, \tilde \mu \geq 0 \notag
                \end{align}
                \item \textbf{Compute Frank Wolfe gap.}  $\quad g \longleftarrow \left(\hat\lambda^\top q\ -\ \hat \mu \cdot B \right) - \left(\lambda^\top q\ -\ \mu \cdot B\right)$\,.
                \item \textbf{Frank Wolfe Update.} If $g > \epsilon$: $$\lambda \longleftarrow (1 - \eta_t) \cdot \lambda + \eta_t \cdot \hat \lambda\,;\quad \mu \longleftarrow (1 - \eta_t) \cdot \mu + \eta_t \cdot \hat \mu\,.$$
            \end{itemize}
            \item[\textbf{Return:}] Dual variables $(\lambda, \mu)$; pricing policies $P_\ell^*(\lambda_\ell) = \argmin_{P_\ell \in \pp_\ell} S_\ell(P_\ell) + Q_\ell(P_\ell) \cdot \lambda_\ell$.
    \end{algorithmic}
\end{algorithm}

We are now ready to describe and motivate our algorithm. Algorithm~\ref{alg:dfw} applies the Frank Wolfe (Conditional Gradient) algorithm to the dual optimization problem~\eqref{eq:combined-dual-problem}. It can be viewed as a price-discovery process that coordinates the spot marketplace and the contract capacity. The dual variable $\lambda_\ell$ is a shadow price for failing to procure a carrier for load $\ell$ on the spot marketplace, and $\mu$ is a shadow price for consuming one unit of contract capacity. In each iteration, the algorithm runs a simple tâtonnement: given the current prices $(\lambda,\mu)$, the spot side responds by optimally pricing each load as if it faced a terminal penalty $w_\ell = \lambda_\ell$, and the pricing oracle returns the resulting non-procurement probabilities $q_\ell(\lambda_\ell)$. These $q_\ell$ represent the residual demand that the contract system must absorb at the current prices.

On the contract side, the algorithm then solves the dual LP to find a new pair $(\hat\lambda,\hat\mu)$ that are the shadow prices for covering the demand $q(\lambda)$ with contracts. In particular, note that the dual LP in Algorithm~\ref{alg:dfw} is exactly the dual of the fractional contract-assignment LP with load profile $q(\lambda)$:
\begin{align*}
    \min\quad  &\sum_{i \in \L} a_i \cdot z_i\\
    \text{s.t.}\quad 
    & z_i\ +\ \sum_{A: i \in A} x_{A} \geq q_i(\lambda_i) \qquad \forall\ i \in \L\\
    & \sum_{A\in \A} x_{A} \leq B\\
    &z_i, x_{A} \in \R_+\,.
\end{align*}
Thus, the Frank Wolfe update acts as a bargaining protocol over the shadow prices of the loads. The spot side takes the current shadow prices $\lambda$ and implicitly states that it is economical to clear only a fraction $1 - q_\ell(\lambda_\ell)$ of each load $\ell$ on the spot marketplace; the remainder $q(\lambda)$ must be absorbed by contracts. The contract side replies that, given this residual demand and the limited capacity $B$, the implied shadow prices should be $(\hat\lambda,\hat\mu)$. The Frank Wolfe step then interpolates between $(\lambda,\mu)$ and $(\hat\lambda,\hat\mu)$, gradually reconciling these two views. At convergence, the leftover demand $q(\lambda)$ and the shadow prices $\lambda$ are approximately consistent: the shadow prices are (approximately) optimal for the leftover demand.

\begin{theorem}\label{thm:runtime}
    Algorithm~\ref{alg:dfw} terminates after at most $\frac{\max_\ell \beta_\ell \cdot A_\max^2}{\epsilon}$ iterations of the \textbf{While} loop.
\end{theorem}

Theorem~\ref{thm:runtime} shows that the shadow price tâtonnement converges in a number of iterations that does not depend on the number of loads or the size of the pricing MDPs. The bound scales as $\beta \cdot A_{\max}^2 / \epsilon$, where $\beta = \max_\ell \beta_\ell$ summarizes how sensitive the non-procurement probabilities are to changes in the penalties, and $A_{\max}$ captures the alternative cost. Each iteration consists of one call to the pricing oracle for each load, which can be done fully in parallel across loads, and one solve of the dual LP over $(\lambda,\mu)$, which only depends on the contract structure and not on the internal details of the pricing MDPs. Thus, the outer algorithm scales gracefully with market size: the iteration count is independent of the number of loads and contracts, while the per-iteration work decomposes cleanly into a parallel pricing phase and a contract LP phase.

Next, we provide a bound the performance loss Algorithm~\ref{alg:dfw} as compared to the computationally-intractable optimal global DP solution $\opt$. Let $\alg$ denote the performance of the pricing policy computed by Algorithm~\ref{alg:dfw}:
\begin{align*}
    \alg\ \coloneqq\ \left\{\sum_{\ell \in \L} S_\ell(P_\ell^*(\lambda_\ell)) \right\}\ +\ \E_{X_\ell \sim Q_\ell(P_\ell^*(\lambda_\ell))}[C(X)]\ =\ \left\{\sum_{\ell \in \L} s_\ell(\lambda_\ell)\right\}\ +\ \E_{X_\ell \sim q_\ell(\lambda_\ell)}[C(X)]\,.
\end{align*}

\begin{theorem}\label{thm:performance-guarantee}
    For tolerance $\epsilon > 0$, Algorithm~\ref{alg:dfw} satisfies
    \begin{align*}
        \alg\ -\ \opt \leq A_\max \cdot \left(\frac{\sqrt{L}}{2}\ +\ 2 \cdot \gamma \cdot \log_2(4L) \right)\ +\ \epsilon\,.
    \end{align*}
\end{theorem}

Theorem~\ref{thm:performance-guarantee} shows that the expected cost of the Dual Frank Wolfe policy exceeds the optimum by at most $A_{\max} (\sqrt{L}/2 + 2\gamma \log_2(4L)) + \epsilon$. The $\sqrt{L}$ term is the price of replacing the discrete contract assignment by its fluid relaxation, and the $2\gamma \log_2(4L)$ term reflects the non-substitutability in the contract system captured by Assumption~\ref{assum:contract-subs}. Importantly, the algorithm never needs to know $\gamma$; it only appears in the bound. In large markets, $\opt$ scales linearly with $L$: each load has to be executed somehow, so even under optimistic assumptions the DFM spends on the order of a constant per load, yielding $\opt = \Theta(L)$. By contrast, the error terms grow strictly sublinearly when contract capacity is reasonably substitutable (for example, when $\gamma = \tilde O(\sqrt{L})$ in structured contract systems such as lane-level or regional contracts). In that regime, the relative performance loss of Dual Frank Wolfe behaves like $O(1/\sqrt{L}) + O(\gamma \log_2(L)/L)$ and therefore vanishes as the market grows, while $\epsilon$ can be made arbitrarily small by tightening the stopping tolerance. In other words, in the large-market limit the Dual Frank Wolfe policy delivers near-optimal coordination between spot and contract supply in a computationally efficient and modular manner.

\subsection{Contract Capacity is Not Wasted}\label{subsec:contract_utilization}

In this section, we show that the policy induced by Dual Frank Wolfe uses nearly all of the available contract capacity with high probability. 
The tension between spot spend and contract utilization is a major operational driver for digital freight marketplaces. Contracted capacity is largely a fixed commitment once agreements are signed, and under-utilizing it is not just an accounting inefficiency: it creates operational friction (planners see idle capacity), commercial friction (carriers question volumes and renewals), and internal friction (stakeholders interpret low utilization as waste even when overall cost looks reasonable). A coordination policy that reduces spot spend but leaves a meaningful fraction of contracts unused is therefore hard to justify and even harder to implement. We now show that Dual Frank Wolfe avoids this failure mode: the policy it induces drives the system toward near-full contract usage, up to small tolerance and concentration effects.

\begin{proposition}\label{prop:contract_utilization}
    Let $(\lambda, \mu)$ be the iterates returned by \Cref{alg:dfw}. Then the residual load profile $q(\lambda)$ left for contracts satisfies
    $\B(q(\lambda)) \geq B - \epsilon \cdot (\beta/\nu)$,
    where $\beta \coloneqq \max_{\ell \in \L} \beta_\ell$.
\end{proposition}

\begin{corollary}[High-probability contract utilization]\label{cor:contract_utilization_hp}
    Let $(\lambda,\mu)$ be the output of \Cref{alg:dfw}, and let $X \sim q(\lambda)$ be the random set of loads not accepted on the load board under the induced pricing policy. Let
    $U(X) \coloneqq \min\{\B(X),\,B\}$
    denote the number of contracts used by the subsequent contract-assignment step. Then for any $\delta \in (0,1)$,
    \[
        U(X) \ \ge\ B - \epsilon \cdot (\beta/\nu)\ -\ \sqrt{\frac{L \ln(1/\delta)}{2}}
    \]
    with probability at least $1-\delta$.
\end{corollary}

The economic intuition mirrors the tâtonnement interpretation of Dual Frank Wolfe. The contract capacity impacts the shadow price $\mu$: when residual demand is high and contracts are scarce, $\mu$ rises, which in turn permits larger load-level penalties $\lambda_\ell$ and pushes the spot side to price more aggressively so that more carriers accept loads on the load board. When residual demand is low and contracts would sit idle, the opposite happens: the capacity price $\mu$ collapses, which forces the load penalties $\lambda_\ell$ down, and the spot side becomes less aggressive, i.e., more loads remain unaccepted on the load board and naturally flow to contracts. In other words, unused contract capacity is not a stable outcome of the iterative shadow price updates under \Cref{alg:dfw}---slack capacity drives $\mu$ and $\{\lambda_\ell\}_\ell$ down, and that signal shifts demand volume back toward contracts until the two sides approximately clear. This is precisely the operational behavior one wants from a coordination layer because it ensures that the committed contract capacity is utilized well. 

Contrast this with the LBA algorithm (Section~\ref{sec:naive-alg}), which explicitly targets 100\% contract utilization in its design philosophy, and ends up suffering higher costs as a consequence because it ignores the preferences of the carriers in the spot marketplace. Our DFW algorithm on the other hand aims to minimize cost through the coordination of spot pricing and contract assignment, while simultaneously attaining a near 100\% contract utilization as a desirable by-product. This difference in performance between the LBA and DFW algorithms highlights the power of flexibility in contract utilization: just a small amount of under-utilization risk can be leveraged to dramatically lower cost.

\section{Applications}\label{sec:numerics}

In this section, we discuss various types of contracts, characterize their $\gamma$ values, and numerically evaluate the performance of Dual Frank Wolfe ( Algorithm~\ref{alg:dfw}) on synthetic and real-world data. As a comparison benchmark, we use the Load Bifurcation Algorithm (LBA) discussed in Section~\ref{sec:naive-alg}. Relative savings are reported as
$\big(\mathbb{E}[\text{Cost(LBA)}]-\mathbb{E}[\text{Cost(DFW)}]\big)/\mathbb{E}[\text{Cost(LBA)}]$.

\subsection{Lane-Level Contracts}

We start with the simplest type of contract capacity: lane-level contracts. These are commitments from carriers to ship a certain number of loads on a particular lane (origin-destination pair) at a specified frequency (daily/weekly/monthly). As discussed in Section~\ref{sec:model}, this becomes $\A = \{\{\ell\} \mid \ell \in \L\}$ in our model because one can decompose the problem along each lane. For lane level contracts, we have $\B(X) = \sum_\ell X_\ell$, and therefore $\gamma = 0$. We note this simple result in the following proposition.

\begin{proposition}\label{prop:lane-level-gamma}
    Lane-level contracts satisfy \Cref{assum:contract-subs} with $\gamma = 0$.
\end{proposition}

\begin{figure}[t]
    \centering
    \includegraphics[width=0.9\textwidth]{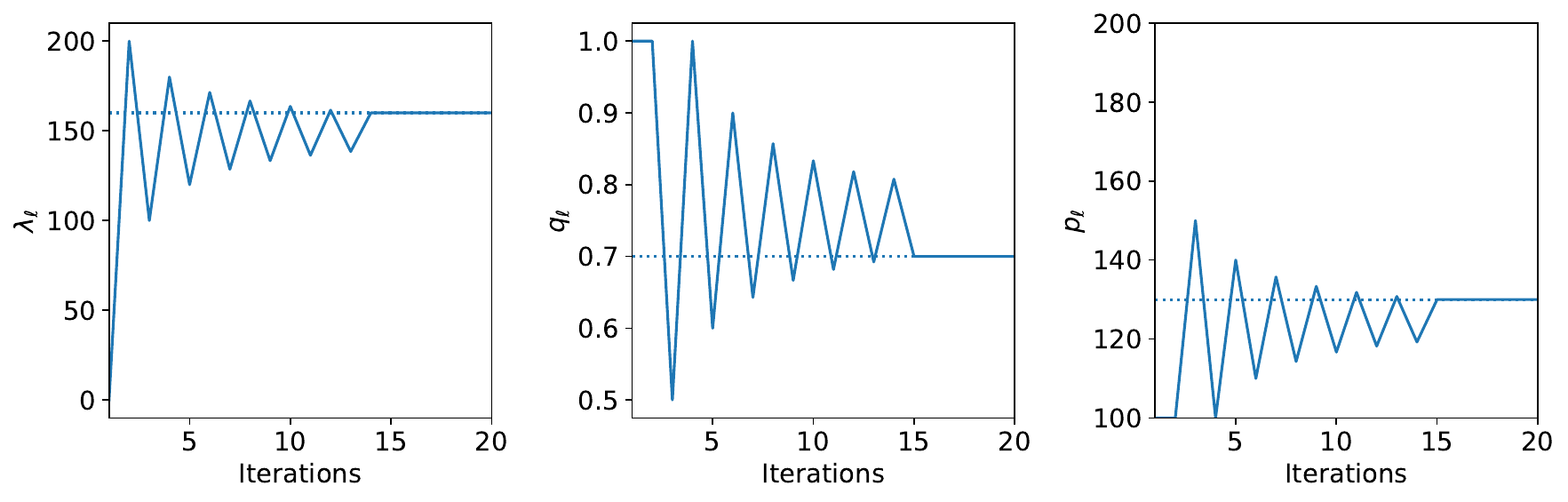}
    \caption{Iterates of the Dual Frank Wolfe algorithms for Example~\ref{alg:dfw}.}
    \label{fig:example_sim}
\end{figure}

To build intuition, we start with the simple setup of Example~\ref{example:naive-bad}. Briefly, the setup is given by ${L = 1000}$ loads; $B = 700$ contracts; static pricing with minimum-acceptable-price uniformly distributed on $[100, 200]$ for every load; alternate cost $a_\ell = 300$. Upon implementing Dual Frank Wolfe, we find that the dual prices $\lambda_\ell$ converges to exactly $160$ in 15 iterations of the \textbf{While} loop, yielding a Frank Wolfe error of $g = 0$. This in turn implies $q_\ell(\lambda_\ell) = 0.3$ and a spot price of $p_\ell = 130$ for all loads $\ell \in \L$, which is exactly the alternate solution that drastically improved upon the naive LBA in Example~\ref{example:naive-bad} and dropped costs by over 30\%. Figure~\ref{fig:example_sim} plots the evolution of the dual variables $\lambda_\ell$, the non-procurement probability $q_\ell$, and the spot price $p_\ell$ under the  Dual Frank Wolfe algorithm.

\begin{figure}[t]
    \centering
    \includegraphics[width=0.9\textwidth]{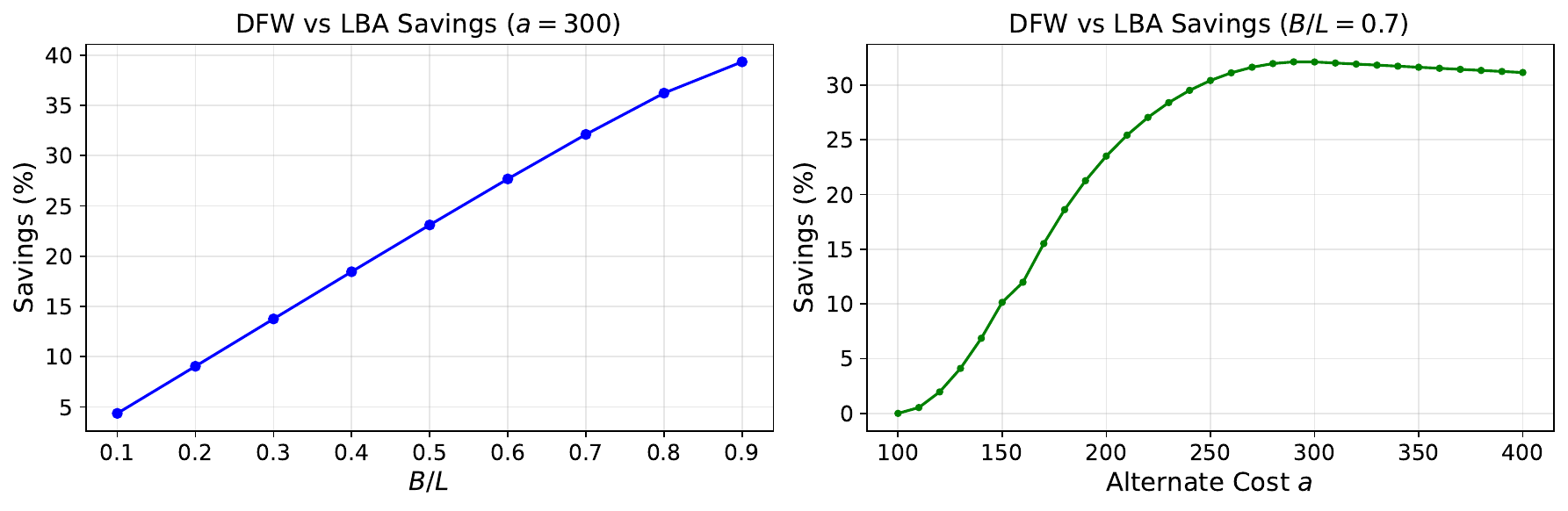}
    \caption{Relative savings of DFW against LBA in Example~\ref{alg:dfw} as a function of the alternate-cost and contract capacity.}
    \label{fig:alternate-cost}
\end{figure}

\Cref{fig:alternate-cost} depicts the impact of the the contract capacity $B$ and the alternate cost $a_\ell$ on the relative savings of the DFW algorithm against LBA. First, we hold the alternate cost $a_\ell = 300$ fixed and vary the contract capacity $B$. When $B\ll L$, both DFW and LBA are forced to procure almost all loads via the spot marketplace and alternate channels, yielding similar costs. As the contract capacity increases, DFW is able to leverage it to execute the loads which are not desired on the spot marketplace and therefore carry a premium, leading to substantial savings compared to LBA. Next, we hold the contract capacity $B = 700$ fixed and vary the alternate cost $a_\ell$. When $a_\ell = 100$, neither algorithm procures anything on the spot marketplace, instead relying entirely on contracts and alternate channels. As $a_\ell$ increases and the alternate channel becomes more expensive, DFW's relative savings increase because, unlike LBA, it can fully utilize the spot and contract capacity without substantially relying on the alternate channel.  

Finally, we modify this example to account for realistic minimum-acceptable-price distributions. In particular, based on data for a particular lane from a major Digital Freight Marketplace (DFM), we trained a sigmoid curve\footnote{Here a sigmoid curve with parameters $(k,x_0)$ is defined as $f(x) = 1/(1 + e^{-k\cdot (x-x_0)})$.} on acceptance data of carriers in the spot marketplace as a function of the price offered. This sigmoid curve captures the probability of a load getting accepted by a carrier on the spot marketplace as a function of the price. Importantly, we would like to disclose that this logistic curve was trained on rescaled data and the scale does not reflect the prices offered by the DFM. 

\begin{figure}[t]
    \centering
    \begin{subfigure}{0.32\textwidth}
        \centering
        \includegraphics[width=\linewidth]{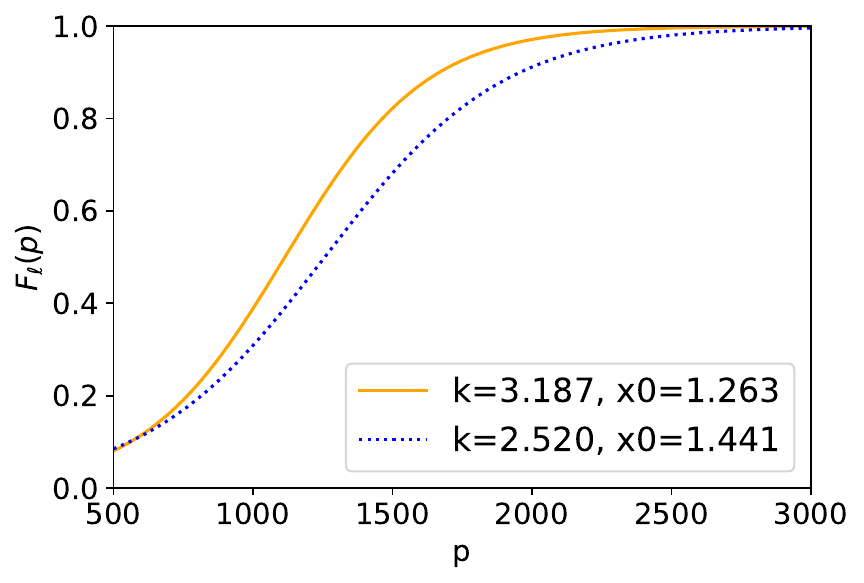}
    \end{subfigure}
    \hfill
    \begin{subfigure}{0.32\textwidth}
        \centering
        \includegraphics[width=\linewidth]{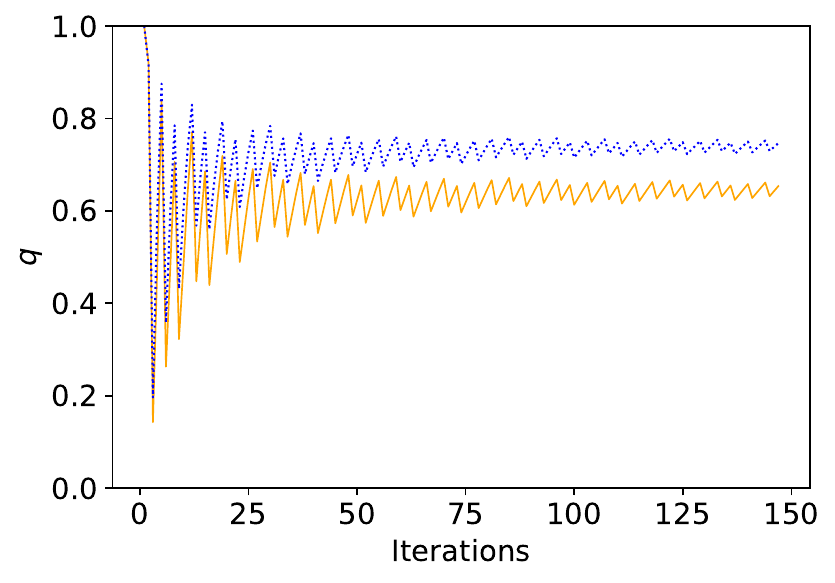}
    \end{subfigure}
    \hfill
    \begin{subfigure}{0.32\textwidth}
        \centering
        \includegraphics[width=\linewidth]{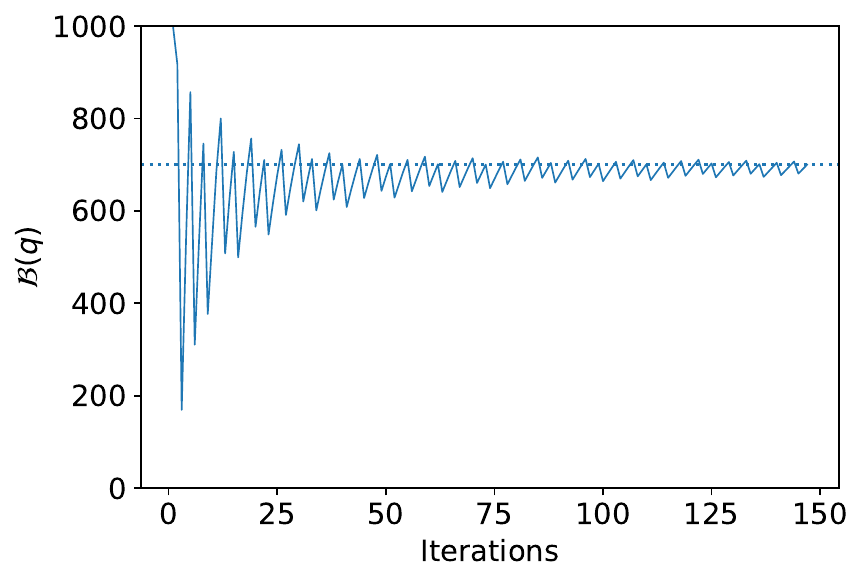}
    \end{subfigure}

    \caption{Lane level contracts with sigmoid spot marketplace pricing curves.}
    \label{fig:sigmoid_lane}
\end{figure}

We split the data into two segments based on the lead time of the load and Figure~\ref{fig:sigmoid_lane} depicts the corresponding sigmoid curves. Moreover, to account for the scale of the problem, we set the alternate cost to be $a_\ell = 5000$ for all loads. With a tolerance of $\epsilon = 100$, the DFW algorithm converged in 147 iterates and the resulting policy yielded a total expected cost of 309,813. This was $\approx$50\% lower than the total expected cost incurred by LBA (597,380). Figure~\ref{fig:sigmoid_lane} shows the convergence behavior of the DFW algorithm on this example.

\subsection{Bipartite Round-Trip contracts}\label{sec:app-bipartite}

Beyond lane-level volume commitments, a second contract structure that appears naturally in long-range freight networks is
driven by the return-to-domicile requirement. In cross-country and inter-regional operations, drivers strongly prefer
tours that start and end at a home base (domicile) within the contracted work shift. Operationally, this brings the driver and equipment back home. Therefore, many contracts explicitly guarantee a return to the
home domicile~\footnote{\url{https://relay.amazon.com/blog/amazon-trucking-contracts}; accessed February 2026.}. As a result, the efficient use of contracted capacity hinges on the ability to \emph{pair} complementary
loads into out-and-back sequences: if a driver takes an $O\to D$ load, the natural next step is to find a compatible
$D\to O$ backhaul to complete the round trip. When such pairing is not available, the driver can still execute a single
load, but the tour typically requires an empty repositioning leg (deadhead) to return home, which is precisely the type
of waste that contracts and routing teams aim to avoid. We capture this round-trip pairing structure via a bipartite graph.

We consider the simplest two-domicile model with locations $O$ and $D$. There are $B$ contracts with drivers domiciled
in $O$. Each load $\ell\in\L$ either starts in $O$ and ends in $D$, or starts in $D$ and ends in $O$; abusing notation,
we denote the two sets of directional loads by $O$ and $D$, respectively. Each contract can be assigned either an
individual load (by pairing it with an empty movement), or a pair
of loads that forms a round trip, i.e., one load from $O$ and one from $D$. Feasibility of pairing depends on operational
constraints such as appointment windows, trailer compatibility, and shift-length feasibility. We encode these constraints
by a bipartite graph $G=(O,D,E)$, where an edge $(\ell_O,\ell_D)\in E$ indicates that the pair can be executed by a single contracted driver while satisfying the return-to-domicile requirement. Thus, the feasible family $\A$ consists of all
singletons $\{\ell\}$ and all feasible pairs $\{\ell_O,\ell_D\}$ corresponding to edges in $E$.

We calibrated this model using contract-feasibility data from a major Digital Freight Marketplace, constructing a graph
with $|O|=|D|=50$ loads and an empirical edge set $E$ capturing which opposing-direction loads can be paired into a round trip. On the spot side, we fit a logistic conversion curve for each load based on historical load-board outcomes. For our simulation runs, we set the alternate-channel cost to be a controlled premium over the median
posted-price level: for $x\in\{0,5,10,20,50\}$, we set $a_\ell=(1+x/100)\cdot p_\ell^{\mathrm{med}}$, where
$p_\ell^{\mathrm{med}}$ is the price at which the fitted acceptance probability is $1/2$. For each value of $x$, we
estimated expected costs via $1{,}000$ Monte Carlo runs. Across all values of $x$, Dual Frank Wolfe converged rapidly: every run
terminated within $40$ Frank Wolfe iterations.


\begin{table}[H]
\centering
\caption{Comparison of DFW and LBA on a real-world bipartite-contract instance.}
\label{tab:bipartite-savings}
\begin{tabular}{rccccc}
\toprule
Alt.\ premium $x$ (\%) & 0 & 5 & 10 & 20 & 50 \\
\midrule
Savings $\pm$ Std. Err. (\%) & \textbf{3.5 $\pm$ 0.1} & \textbf{4.1 $\pm$ 0.1} & \textbf{4.8 $\pm$ 0.1} & \textbf{6.2 $\pm$ 0.1} & \textbf{9.6 $\pm$ 0.1} \\
\bottomrule
\end{tabular}
\end{table}

Table~\ref{tab:bipartite-savings} summarizes the results. Dual Frank Wolfe consistently improves upon the Load
Bifurcation Algorithm, with savings close to $10\%$ when the alternate channel is $50\%$ more expensive than the median spot cost---as is often the case in practice because the alternate channel is a team of human operators manually calling carriers to negotiate prices at short notice. Even a single-percentage cost reduction is operationally significant at the scale of a large DFM because their annual procurement costs routinely run in the millions of dollars, often exceeding a billion dollars for the large firms. These savings are in agreement with our theory: by discovering shadow prices that internalize contract scarcity under the pairing constraints, Dual Frank Wolfe adjusts spot-market prices so that the residual load set remains easier to cover with round-trip contracts. 

To further verify the alignment with theory, we tested how the degree of the bipartite graph affect the level of relative savings. In particular, as the degree of the graph increases, so does the level of substitutability in the contracts. This is exactly what the substitutability parameter $\gamma$ is defined to capture: when each load has many feasible partners, contract capacity can be repacked across many different residual load sets and the system behaves as if contracts were broadly interchangeable (small $\gamma$); when feasible pairings are scarce, contract utilization becomes brittle to the particular loads that remain after the spot marketplace clears (large $\gamma$). The extremes make this intuition sharp. In a $1$-regular graph, every load has essentially a unique compatible partner, so random spot-market realizations frequently “break” potential round trips by removing one endpoint while leaving the other stranded, yielding a large effective slack. At the other extreme, in a complete bipartite graph any $O\to D$ load can be paired with any $D\to O$ load, so the return-to-domicile constraint imposes little combinatorial friction and the corresponding $\gamma$ is small.

To probe this mechanism, we studied instances where the contract-assignment feasibility was captured by random graphs. In particular, we use Erd\H{o}s--R'enyi random bipartite graphs $G(50,50,p)$ where each edge is sampled independently with probability $p$. We set $|O|{=}|D|{=}50$, $B=25$, i.i.d.\ carrier values $\mathrm{Uniform}[100,200]$, and alternate cost $a=300$, and averaged the relative savings of DFW over LBA over 1,000 graph draws and spot realizations. Figure~\ref{fig:savings-vs-p} depicts the results. First, we found that even when the average degree is as low as $5$ (at $p=0.10$), DFW’s relative savings are substantial and exceed $12\%$. Second, these savings plateau around $16.5\%$ by average degree $\approx 10$ (at $p=0.20$), indicating that additional edges yield diminishing returns once residual loads are almost always pairable. Consistent with the substantial savings observed in Table~\ref{tab:bipartite-savings}, the empirical bipartite graph extracted from the real-world Digital Freight Marketplace data had an average degree above $10$, placing it squarely in this high-substitutability regime. Finally, the relative savings are non-monotonic at very low degrees: at $p=0$, there are no round-trip edges and we are back to the lane-level contract setting---DFW delivers about $9\%$ savings, while for very sparse networks with degree close to 1, the round-trip constraint is brittle and DFW can underperform LBA---the substitutability parameter $\gamma$ is large for such graphs because every load has a unique partner and there is little scope for substitution. Overall, except for a narrow window of degrees which lead to graphs that are sparse but not too sparse, DFW generates substantial relative savings, comfortably exceeding 10\% for the realistic values we observed in real-world data.

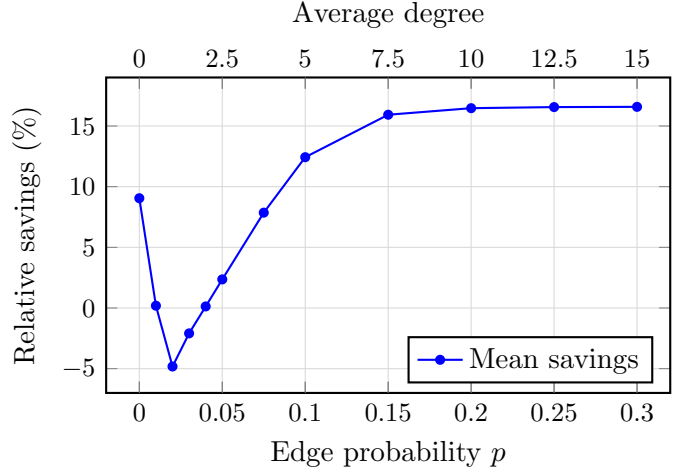
\begin{figure}[t]
\centering
\begin{tikzpicture}
\begin{axis}[
    width=0.55\textwidth,
    height=0.35\textwidth,
    xlabel={Edge probability $p$},
    ylabel={Relative savings (\%)},
    xmin=-0.02, xmax=0.32,
    ymin=-7, ymax=19,
    xticklabel style={/pgf/number format/fixed},
    ytick distance=5,
    extra x ticks={0, 0.05, 0.10, 0.15, 0.20, 0.25, 0.30},
    extra x tick labels={0, 2.5, 5, 7.5, 10, 12.5, 15},
    extra x tick style={
        ticklabel pos=right,
        grid=none,
    },
    extra description/.append code={
        \node[font=\normalsize] at (axis description cs:0.5,1.2) {Average degree};
    },
    grid=major,
    grid style={gray!30},
    mark size=1.5pt,
    thick,
    legend pos=south east,
    tick label style={font=\small},
    label style={font=\normalsize},
    title style={font=\normalsize},
]

\addplot[
    color=blue,
    mark=*,
    error bars/.cd,
    y dir=both,
    y explicit,
] coordinates {
    (0.00,  9.05) +- (0, 0.01)
    (0.01,  0.19) +- (0, 0.155)
    (0.02, -4.82) +- (0, 0.085)
    (0.03, -2.09) +- (0, 0.06)
    (0.04,  0.13) +- (0, 0.065)
    (0.05,  2.36) +- (0, 0.07)
    (0.075,  7.86) +- (0, 0.08)
    (0.10, 12.43) +- (0, 0.06)
    (0.15, 15.93) +- (0, 0.02)
    (0.20, 16.47) +- (0, 0.01)
    (0.25, 16.56) +- (0, 0.01)
    (0.30, 16.58) +- (0, 0.01)
};
\addlegendentry{Mean savings}

\end{axis}
\end{tikzpicture}
\caption{Relative savings of DFW over LBA on random Erd\H{o}s--R\'{e}nyi bipartite graphs $G(50,50,p)$ as a function of edge probability~$p$. 95\% confidence intervals are smaller than the marker size.}
\label{fig:savings-vs-p}
\end{figure}

\subsection{Regional Contracts}\label{sec:app-interval}

We now consider contracted fleet operations in a small-to-medium sized
region like a metropolitan area or special economic zone. Under such contracts, a driver commits to a fixed shift---say, a $13$-hour window
(e.g., 6am--7pm)---and the driver may be routed to execute multiple loads within the region during the shift.
This contract form is especially natural in short-haul and drayage operations, where loads are geographically local
and the key operational constraint is temporal: a driver can only move one load at a time.

To model such contracts, we discretize the shift into hourly slots $t \in \{1,\dots,T\}$, and associate each load
$\ell\in\L$ with an \emph{active interval} $I_\ell \subseteq \{1,\dots,T\}$ representing the hours during which a driver
is effectively occupied by that load (pickup, transit, unloading, and any required dwell). In practice, these active intervals are often padded to account for repositioning time between drop-off and the next pickup (deadhead),
stochastic congestion, and operational slack (e.g., check-in delays). Concretely, one may take a nominal service window
and extend it by a fixed buffer (say 15--45 minutes, rounded to the hourly grid) at the end so that
any schedule that is feasible for the padded intervals is operationally feasible with high probability.

Local geography can be incorporated via \emph{colors}. In particular, the region can be partitioned into a small number of areas, and each load can be assigned a color $c(\ell)\in[K]$ indicating the area in which the load lies. Operationally, drivers frequently prefer staying within an area for a
shift due to familiarity, parking/access constraints, yard rules, and end-of-shift return requirements. We model this
as a \emph{monochromatic} restriction: all loads assigned to a contracted driver must have the same color.
In this colored interval model, a feasible assignment for a single contract is a set of loads with the same color whose
active intervals do not overlap. The resulting contract-requirement function $\B(\cdot)$ admits an explicit form. For integer demand vectors
$X\in\mathbb{Z}_+^{|\L|}$ (allowing multiple copies of each load, as in our model), the minimum number of contracts needed
is determined by the largest hourly overlap within each color class:
\begin{align*}
\B(X)\ \coloneqq\ \sum_{k=1}^K\ \max_{t\in\{1,\dots,T\}}\ \sum_{\ell:\ c(\ell)=k,\ t\in I_\ell} X_\ell.
\end{align*}

\begin{proposition}\label{prop:gamma-interval-graph}
    Regional contracts satisfy \Cref{assum:contract-subs} with $\gamma = \sqrt{L \cdot K \log(T)/2}$.
\end{proposition}

The proposition provides a concrete substitutability parameter $\gamma$ for regional contracted fleets. It shows
that these contracts are broadly substitutable even when the load mix is distributed unevenly across neighborhoods and across the hours of a shift. In particular, the number of contracts required to cover it remains stable: the additional slack needed to accommodate the routing constraints grows only
sublinearly, on the order of $O(\sqrt{L})$. Since the overall performance bound already carries an unavoidable
$O(\sqrt{L})$ term from the fluid approximation, this value of $\gamma$ does not materially increase the regret of the
DFW algorithm.

\section{Conclusion}

This paper studies the operational problem of coordinating committed contract capacity with load-board spot pricing in freight marketplaces. We formalize a benchmark in which a marketplace jointly selects spot pricing policies (over time) and a contract assignment (over combinatorial feasible routes) to minimize total expected procurement cost. We show that the common status-quo load bifurcation approach---committing ex ante which loads are assigned to contracts and which go to spot---can be highly suboptimal because it ignores the preferences of the carriers on the spot marketplace. We propose Dual Frank Wolfe (DFW), a modular coordination layer built around shadow prices for loads and contract capacity. DFW interfaces with existing systems through black-box oracles: a per-load pricing oracle that optimizes the marketplace’s implementable spot policy given a terminal penalty, and a contract-assignment oracle captured through a small dual LP. The resulting tâtonnement-style procedure converges in a number of iterations that does not scale with market size, admits parallel implementation, and yields a policy whose expected cost is provably near-optimal in large markets under a natural contract-substitutability condition. Importantly, DFW also drives the system toward near-full utilization of committed contract capacity, addressing a key operational constraint for deployability.

Numerical experiments on synthetic and semi-synthetic instances calibrated to a major DFM support the theory. Across lane-level and pairing-based round-trip contracts, DFW consistently improves upon the load-bifurcation baseline, with gains that become larger as the alternate channel becomes more expensive and as contract flexibility increases. Overall, the results suggest that shadow-price-based coordination can deliver material procurement savings without requiring a redesign of the spot pricing engine or the contract assignment infrastructure.

This topic of coordinating contract assignment and spot pricing decisions at the time of execution remains very underexplored~\citep{acocella2023research}. Many open problems remain: What is the impact of late-stage load cancellations? How do tender rejections from carriers impact outcomes? How should one account for correlations in spot outcomes for sparse markets? How should one make pricing decisions on the shipper side to determine the load set when the downstream carrier procurement has been coordinated?

{\singlespacing
\bibliographystyle{plainnat}
\bibliography{refs}
}

\allowdisplaybreaks
\newpage
\appendix

\pagenumbering{arabic}\renewcommand{\thepage}{ec \arabic{page}}
\section{Missing Proofs for Section~\ref{sec:alg}}

\subsection{Proof of Lemma~\ref{lemma:smoothness}}

\begin{proof}
    Observe that $r_\ell$ can be written as
    \begin{align*}
        r_\ell(w) = \min_{P \in \pp_\ell} S_\ell(P) + Q_\ell(P) \cdot w\,.
    \end{align*}
    In other words, $w \mapsto r_\ell(w)$ is the minimum of linear functions $w \mapsto S_\ell(P) + Q_\ell(P) \cdot w$. Therefore, the Envelope Theorem applies and we have
    \begin{align*}
        r'_\ell(w) = Q_\ell(P_\ell^*(w)) = q_\ell(w)\,.
    \end{align*}
    Finally, using \Cref{assum:lipschitz}, we get
    \begin{align*}
        |r'_\ell(w) - r'_\ell(\tilde w)| = |q_\ell(w) - q_\ell(\tilde w)| \leq \beta_\ell\cdot |w - \tilde w|\,.
    \end{align*}
    Thus, $r_\ell(\cdot)$ is $\beta_\ell$-smooth.
\end{proof}

\subsection{Proof of Lemma~\ref{lemma:approx-convexity}}

Lemma~\ref{lemma:approx-convexity} is foundational to our results. Its proof proceeds over multiple steps. However, before we dive into the proof, we need to establish some prerequisites.

\begin{lemma}
\label{lem:mon-subadd}
The function $\B : \Z_+^L \to \R_+$ satisfies
\begin{enumerate}
    \item (Monotonicity) If $X,Y \in \Z_+^L$ and $X \le Y$ coordinatewise, then
    \[
    P(X) \;\le\; P(Y).
    \]
    \item (Subadditivity) For all $U,V \in \Z_+^L$,
    \[
    \B(U+V) \;\le\; \B(U) + \B(V).
    \]
\end{enumerate}
\end{lemma}

\begin{proof}
(1) Recall the definition of $\B(X) = \min \{\sum_{A \in \A} x_A \mid \sum_{A: i \in A} x_{A} \geq X_i\ \forall\ i \in \L; x_A \in \Z_+\ \forall\ A \in \A\}$ If $X \le Y$, then any $x$ feasible for $Y$ is also feasible for $X$. Thus the feasible region for $Y$ is contained in that for $X$, and the minimum cost for $X$ cannot exceed that for $Y$.

(2) Let $\{x_A\}_A$ be optimal for $U$ and $\{y_A\}_A$ be optimal for $V$. Define $z_A = x_A + y_A$. Then for each $i \in [L]$,
\[
\sum_{A \ni i} z_A
=
\sum_{A \ni i} x_A + \sum_{A \ni i} y_A
\;\ge\;
U_i + V_i,
\]
so $z$ is feasible for $U+V$. Hence
\[
P(U+V)
\;\le\;
\sum_{A} z_A
=
\sum_A x_A + \sum_A y_A
=
P(U) + P(V).
\]
\end{proof}

\begin{lemma}\label{lemma:sub-additivity}
    For all $d, \tilde d \in \R_+^L$, we have
    \begin{align*}
        \B( d+ \tilde d ) \leq \B( d ) + \B( \tilde d )\,.
    \end{align*}
\end{lemma}

\begin{proof}
Fix $d,\tilde d \in \R_+^L$. For each $i \in [L]$, write
\[
d_i = a_i + b_i,
\quad
\tilde d_i = \tilde a_i + \tilde b_i,
\quad
a_i,\tilde a_i \in \Z_+,\ b_i,\tilde b_i \in [0,1),
\]
and
\[
d_i + \tilde d_i = c_i + \gamma_i,
\quad
c_i := \lfloor d_i + \tilde d_i \rfloor \in \Z_+,\ \gamma_i \in [0,1).
\]
Thus
\[
b_i + \tilde b_i = (c_i - a_i - \tilde a_i) + \gamma_i,
\quad
c_i - a_i - \tilde a_i \in \{0,1\}.
\]

We now construct, on a single probability space, random integer vectors
\[
X \sim d,
\quad
Y \sim \tilde d,
\quad
Z \sim d + \tilde d
\]
such that
\[
Z \;\le\; X + Y
\quad\text{coordinatewise.}
\]

\textbf{Per-coordinate coupling.}
Let $(U_i)_{i=1}^L$ be i.i.d.\ $\mathrm{Uniform}[0,1]$. For each $i$, we distinguish two cases.

\medskip\noindent
\emph{Case 1: $b_i + \tilde b_i \le 1$.}

Set
\[
B_i := \mathbf{1}\{U_i \le b_i\},
\qquad
C_i := \mathbf{1}\{b_i < U_i \le b_i + \tilde b_i\},
\qquad
D_i := \mathbf{1}\{U_i \le b_i + \tilde b_i\}.
\]
Then
\[
B_i \sim \mathrm{Bernoulli}(b_i),
\quad
C_i \sim \mathrm{Bernoulli}(\tilde b_i),
\quad
D_i \sim \mathrm{Bernoulli}(b_i + \tilde b_i) = \mathrm{Bernoulli}(\gamma_i),
\]
and $B_i,C_i,D_i$ are all functions of $U_i$.

Moreover, whenever $D_i = 1$ we have $U_i \le b_i + \tilde b_i$, hence either $U_i \le b_i$ (so $B_i=1$) or $b_i < U_i \le b_i + \tilde b_i$ (so $C_i = 1$). Thus
\[
D_i \;\le\; B_i + C_i.
\]

Here $c_i = a_i + \tilde a_i$ and $\gamma_i = b_i + \tilde b_i$.

\medskip\noindent
\emph{Case 2: $b_i + \tilde b_i > 1$.}

Then $c_i = a_i + \tilde a_i + 1$ and $\gamma_i = b_i + \tilde b_i - 1 \in (0,1)$.

Set
\[
B_i := \mathbf{1}\{U_i \le b_i\},
\qquad
C_i := \mathbf{1}\{U_i \ge 1 - \tilde b_i\},
\qquad
D_i := \mathbf{1}\{1 - \tilde b_i \le U_i \le b_i\}.
\]
Since $b_i + \tilde b_i > 1$, the interval $[1-\tilde b_i, b_i]$ has length $b_i + \tilde b_i - 1 = \gamma_i$, so
\[
B_i \sim \mathrm{Bernoulli}(b_i),
\quad
C_i \sim \mathrm{Bernoulli}(\tilde b_i),
\quad
D_i \sim \mathrm{Bernoulli}(\gamma_i).
\]

For any $U_i \in [0,1]$, we have $B_i + C_i \ge 1$ (the events $\{U_i \le b_i\}$ and $\{U_i \ge 1-\tilde b_i\}$ cover $[0,1]$). Moreover, if $D_i = 1$ then $U_i \in [1-\tilde b_i, b_i]$ and hence $B_i = C_i = 1$, so $B_i + C_i = 2$.

Thus in all cases
\[
B_i + C_i \;\ge\; 1 + D_i.
\]
Recall that here $c_i - a_i - \tilde a_i = 1$, so this is exactly
\[
c_i - a_i - \tilde a_i + D_i \;\le\; B_i + C_i.
\]

\medskip

Now define random integer variables for each $i$ by
\[
X_i := a_i + B_i,
\qquad
Y_i := \tilde a_i + C_i,
\qquad
Z_i := c_i + D_i.
\]

By construction and independence of $(U_i)$ across $i$, we have:

\begin{itemize}
    \item $X \sim d$, i.e.\ $X_i = \lfloor d_i \rfloor + \mathrm{Bernoulli}(d_i - \lfloor d_i \rfloor)$ independently.
    \item $Y \sim \tilde d$ in the same way.
    \item $Z \sim d + \tilde d$.
\end{itemize}

Furthermore, from the inequalities above we obtain, in both cases,
\[
c_i - a_i - \tilde a_i + D_i \;\le\; B_i + C_i,
\]
so
\[
Z_i
= c_i + D_i
\le a_i + \tilde a_i + B_i + C_i
= X_i + Y_i,
\]
for each $i$. Hence
\[
Z \;\le\; X + Y
\quad\text{coordinatewise.}
\]

\textbf{Putting it together.} By monotonicity of $\B$ on $\Z_+^L$ (Lemma~\ref{lem:mon-subadd}),
\[
\B(Z) \;\le\; \B(X+Y)
\quad\text{almost surely},
\]
and by subadditivity (Lemma~\ref{lem:mon-subadd}) again,
\[
\B(X+Y) \;\le\; \B(X) + \B(Y)
\quad\text{almost surely}.
\]
Therefore
\[
\B(Z) \;\le\; \B(X) + \B(Y)
\quad\text{almost surely}.
\]

Taking expectations and using the definition of $\B(\cdot)$ completes the proof.
\end{proof}

\begin{lemma}\label{lemma:equivalent-sampling}
    For all $d \in \R_+^L$, we have
    \begin{align*}
        \B\left( \frac{d}{2}\right) = \E_{X \sim d}\left[ \B\left( \frac{X}{2} \right) \right]\,.
    \end{align*}
\end{lemma}
\begin{proof}
 Observe that it suffices to show that $Y \sim X/2$ for $X \sim d$ is supported on ${\{\floor{d/2}, \floor{d/2} + 1\}}$ because $\E[Y] = d/2$ and there is a unique coordinate-wise-independent distribution supported on ${\{\floor{d/2}, \floor{d/2} + 1\}}$ with that mean.

 Consider a particular coordindate $i \in [L]$. Then, $X \sim d$ satisfies $X_i \in \{\floor{d_i}, \floor{d_i} + 1\}$. Consider the following two cases:
 \begin{itemize}
     \item $\floor{d_i}$ is odd. In this case, $\floor{d_i/2} \leq \floor{d_i}/2 \leq (\floor{d_i} + 1)/2 = \floor{d_i/2} + 1$. Therefore, $X_i/2$ is supported on ${\{\floor{d_i/2}, \floor{d_i/2} + 1\}}$.

     \item $\floor{d_i}$ is even. In this case, $\floor{d_i/2} = \floor{d_i}/2 \leq (\floor{d_i} + 1)/2 \leq \floor{d_i/2} + 1$. Once again, $X_i/2$ is supported on ${\{\floor{d_i/2}, \floor{d_i/2} + 1\}}$. 
 \end{itemize}
 Thus, we have established that $X$ is supported on ${\{\floor{d/2}, \floor{d/2} + 1\}}$, as desired.
\end{proof}

We are now ready to prove Lemma~\ref{lemma:approx-convexity}.

\begin{proof}[Proof of Lemma~\ref{lemma:approx-convexity}]
    We begin by showing that $\B(\cdot)$ is $\gamma$-approximately midpoint convex on $\R_{+}^L$: for all $d, \tilde d \in \R_+^L$, we have
    \begin{align}\label{eq:inter-midpoint}
        \B\left( \frac{d + \tilde d}{2} \right) \leq \frac{\B(d) + \B(\tilde d)}{2}\ +\ 2 \cdot \gamma\,.
    \end{align}
    Now, we can leverage result from the theory of the stability of functional equations in several variables. In particular, Theorem~8.9 of \citet{hyers2012stability} in combination with \eqref{eq:inter-midpoint} implies that $\B(\cdot)$ is $4\gamma$-approximately convex on $\R_{++}^L$, i.e., for all $d, \tilde d \in \R_{++}^L$, we have
    \begin{align*}
        \B(\zeta \cdot d + (1 - \zeta) \cdot \tilde d) \leq \zeta \cdot \B(d) + (1 - \zeta) \cdot \B(\tilde d)\,.
    \end{align*}
    As $\B(\cdot)$ is a multi-linear extension and therefore continuous, we get that $\B(\cdot)$ is $4\gamma$-approximately convex on $\R_+^L$. This allows us to apply Theorem~8.3 (where $q_n \leq \log_2(L) + 2$) of \citet{hyers2012stability}, which yields the existence of the desired convex function $g$. Therefore, it suffices to prove \eqref{eq:inter-midpoint} to establish the lemma and we do so in the remainder.

    Consider $d, \tilde d \in \R_+^L$. 
    First, we apply Lemma~\ref{lemma:sub-additivity} to get
    \begin{align}\label{eq:inter-approx-convexity}
        \B\left( \frac{d + \tilde d}{2} \right) \leq \B\left( \frac{d}{2}  \right) + \B\left( \frac{\tilde d}{2} \right)\,.
    \end{align}
     Next, we apply Lemma~\ref{lemma:equivalent-sampling} to get
     \begin{align*}
         \B\left( \frac{d + \tilde d}{2} \right) \leq \E_{X \sim d}\left[\B\left( \frac{X}{2}  \right) \right] + \E_{\tilde X \sim \tilde d}\left[\B\left( \frac{\tilde X}{2}  \right) \right]\,.
     \end{align*}

    
    Finally, Assumption~\ref{assum:contract-subs} implies
     \begin{align*}
         \B\left( \frac{d + \tilde d}{2} \right) \leq \frac{\E_{X\sim d}[\B(X)] + \E_{\tilde X \sim \tilde d}[\B(\tilde X)]}{2} + 2 \cdot \gamma = \frac{\B(d) + \B(\tilde d)}{2} + 2 \cdot \gamma\,,
     \end{align*}
     thereby establishing \eqref{eq:inter-midpoint} and completing the proof.
\end{proof}

\subsection{Proof of Lemma~\ref{lemma:fluid-error}}

\begin{proof}
    Consider a feasible solution $(p,d)$ of $\F(q)$. First, observe that $(p, \min\{d, 1\})$ is also a feasible solution because $q \in [0,1]$. Therefore, without loss of generality, we can assume $d_\ell \in  [0,1]$ for all $\ell \in \L$. As a consequence, we get $\B(d) = \E_{X \sim d}[\B(X)] \leq B$.

    Next, note that $X \mapsto \B(X)$ has bounded differences: adding or removing one load does not change the required number of contracts by more than 1. Concretely, if $X$ and $\tilde X$ differ by at most one coordinate, then
    \begin{align*} 
        |\B(X) - \B(\tilde X)| \leq 1\,.
    \end{align*}
    Thus, we can apply the Efron-Stein inequality to get $\text{Var}_d(\B(X)) = L/2$. Combining this with Jensen's inequality yields:
    \begin{align}\label{eq:inter-fluid-error}
        \E_{d}[(\B(X) - B)^+] \leq \E_{d}[|\B(X) - \B(d)|] \leq \sqrt{\text{Var}(\B(X))} \leq \frac{\sqrt{L}}{2}\,.
    \end{align}
    
    Consider the following coupling of $X \sim d$ and $Y \sim q$: simulate an independent standard uniform random variable $U_\ell$ for each load $\ell \in \L$ and define
    \begin{align*}
        X_\ell = \mathds{1}(u_\ell \leq d_\ell);\quad Z_\ell = \mathds{1}(d_\ell < u_\ell \leq d_\ell + p_\ell);\quad Y_\ell = X_\ell + Z_\ell\,.
    \end{align*}

    We can now define our solution to the contract assignment problem $C(Y)$. Let $x^* \in \Z_+^{|\A|}$ be the contract assignment which satisfies $\sum_{A \in \A} x^*_A = \B(X)$ and $\sum_{A: \ell \in A} x^*_A \geq X_\ell$. Thus, there exist $x \leq x^*$ and $p \in \Z_+^L$ such that $\sum_{A \in \A} x_A = \min\{B, \B(X)\}$ and
    \begin{align*}
        z_\ell = Z_\ell + \begin{cases}
            1 &\text{if}\ \ell \in A\ \text{s.t.}\ x^*_A > x_A\\
            0 &\text{otherwise}\,.
        \end{cases}    
    \end{align*}
    As a consequence, we get $a^\top(z_\ell - Z_\ell) \leq \sum_{A \in A} (\sum_{\ell \in A} a_\ell) \cdot (x^*_a - x_A) \leq A_\max \cdot (\B(X) - B)^+$. Therefore, $(z,x)$ is a feasible solution of $C(Y)$ because $Y_\ell = X_\ell + Z_\ell$ and it has a objective value
    \begin{align*}
        \sum_{\ell \in \L} a_\ell \cdot z_\ell \leq a^\top Z + A_\max \cdot (\B(X) - B)^+\,.
    \end{align*}
    
    Since, we defined these feasible solutions $(z,x)$ for an arbitrary $Y \sim q$, we get
    \begin{align*}
        a^\top p\ +\ \E_{X \sim d}[A_\max \cdot(\B(X) - B)^+] = \E_{(X,Y,Z)}[a^\top Z + A_\max \cdot(\B(X) - B)^+] \geq \E_{Y \sim q}[C(Y)]\,.
    \end{align*}
    Finally, to conclude the proof, use the bound $\E_{X \sim d}[(\B(X) - B)^+] \leq \sqrt{L}/2$ from \eqref{eq:inter-fluid-error}.
\end{proof}

\subsection{Proof of Lemma~\ref{lemma:dual-problem}}

\begin{proof}
    First, observe that if $\lambda_\ell > a_\ell$ for any $\ell \in \L$, then $p_\ell \to \infty$ yields $\D(\lambda, \mu \mid q) = -\infty$. Therefore, we must have $\lambda_\ell \leq a_\ell$ for all $\ell \in \L$ for $\D(\lambda, \mu \mid q)$ to be finite. In this case, $\min_{p\geq 0} (a - \lambda)^\top p = 0$.
    
    Similarly, if $\sum_{\ell \in A} \lambda_\ell > \mu$ for some $A \in \A$, then we can set $d_\ell = n$ for $\ell \in A$ and zero other wise, and consider $n \to \infty$. In this case, we get $\B(d) \leq n$, and therefore $\lim_{n \to \infty} \mu \cdot \B(d) - \lambda^\top d = \lim_{n \to \infty} (\mu - \sum_{\ell \in A} \lambda_\ell) \cdot n = -\infty$. Therefore, we must have $\sum_{\ell \in A} \lambda_\ell \leq \mu$ for all $A \in \A$ in order to ensure that $\D(\lambda, \mu \mid q)$ is finite. In this case, we have
    \begin{align*}
        \min_{d \geq 0}\ \mu \cdot \B(d) - \lambda^\top d = \min_{d \geq 0}\ \E_{X \sim d}[\mu \cdot \B(X) - \lambda^\top X] = \min_{X \in \Z_+^L} \mu \cdot \B(X) - \lambda^\top X \geq 0\,,
    \end{align*}
    where the last inequality follows by assigning $X$ to $\B(X)$ contracts and noting that $\sum_{\ell \in A}\lambda_\ell \leq \mu$ for each contract. Plugging these observations into \eqref{eq:dual_def} completes the proof.
\end{proof}

\subsection{Proof of \Cref{thm:runtime}}
\begin{proof}[Proof of \Cref{thm:runtime}]
    Observe that Algorithm~\ref{alg:dfw} implments Frank Wolfe iterates on the dual concave maximization problem described in \eqref{eq:combined-dual-problem}. We endow the space of dual variables $(\lambda, \mu)$ with the $\ell_\infty$ norm. Therefore, the diameter of the feasibility set is given by $D = A_\max$. Moreover, the gradients are measured in the $\ell_1$ (dual of $\ell_\infty$) norm, and thus Lemma~\ref{lemma:smoothness} implies that $\lambda \mapsto \sum_{\ell \in \L} r_\ell(\lambda_\ell)$ is $(\max_\ell \beta_\ell)$-smooth. The runtime bound follows from Theorem~2.2 of \citet{braun2022conditional}.
\end{proof}

\subsection{Proof of \Cref{thm:performance-guarantee}}
\begin{proof}[Proof of \Cref{thm:performance-guarantee}]
    
    Let $(\lambda, \mu)$ be the iterates with which Algorithm~\ref{alg:dfw} terminates. The proof is split into multiple steps: we first upper bound the cost of $\alg$, then lower bound the cost of $\opt$, and finally put together the two bounds. To simplify notation, we set $q_\ell = q_\ell(\lambda_\ell)$ for all $\ell \in \L$.

    \textbf{Step 1 (Upper bound on algorithm's performance). } Applying Lemma~\ref{lemma:fluid-error} yields
    \begin{align*}
        \alg\ \leq\ \sum_{\ell \in \L} s_\ell(\lambda_\ell)\ +\ \F(q)\ +\ A_\max \cdot \frac{\sqrt{L}}{2}\,.
    \end{align*}
    
    Next, we show that approximate strong duality holds and bound $\F(q)$ using the dual optimization problem derived in Lemma~\ref{lemma:dual-problem}. To do so, we leverage Lemma~\ref{lemma:approx-convexity} to get
    \begin{align*}
        \F(q) \quad \leq \quad \min_{p,d} \quad &a^\top p\\
    \text{s.t.} \quad &p + d \geq q\\
    & g(d) \leq B - 2 \gamma \log_2(4L)\\
    &d,\ p\ \in \R_+^L 
    \end{align*}
    where we have used the fact that $\B(d) \leq g(d) + 4 \cdot \gamma$ for all $d \in \R_+^L$. As $g$ is convex, the RHS optimization problem satisfies strong duality due to Slater's condition (pick a $d$ small enough and $p$ small enough to get an strictly interior point). Therefore, we get
    \begin{align*}
       \F(q) \quad &\leq \quad \max_{\tilde \lambda, \tilde \mu \geq 0}\ \min_{p,d \geq 0}\ (a - \tilde\lambda)^\top p + \left(\tilde  \mu \cdot g(d) - \tilde \lambda^\top d \right) - \tilde \mu \cdot (B - 2  \gamma \log_2(4L)) + \tilde \lambda^\top q\\
       &\leq \quad \max_{\tilde \lambda, \tilde \mu \geq 0}\  \min_{p,d \geq 0}\ \ (a - \tilde\lambda)^\top p +\ \left(\tilde  \mu \cdot \B(d) - \tilde \lambda^\top d \right) - \tilde \mu \cdot (B - 2 \gamma \log_2(4L)) + \tilde \lambda^\top q\\
       &\leq \quad \max_{\tilde \lambda, \tilde \mu \geq 0}\ \tilde \lambda^\top q - \tilde \mu \cdot (B - 2 \gamma \log_2(4L)) \quad \quad \text{s.t.}\quad \left\{\sum_{\ell \in A} \tilde \lambda_\ell \leq \tilde \mu\quad \forall\ A \in \A \,;\ \tilde \lambda\leq a \right\}\\
       &\leq \quad 2 \gamma \log_2(4L) \cdot A_\max\ +\  \max_{\tilde \lambda, \tilde \mu \geq 0}\ \tilde \lambda^\top q - \tilde \mu \cdot B \quad \quad \text{s.t.}\quad \left\{\sum_{\ell \in A} \tilde \lambda_\ell \leq \tilde \mu \quad \forall\ A \in \A\,;\ \tilde \lambda\leq a \right\}\,,
    \end{align*}
    where the second-last inequality follows from Lemma~\ref{lemma:dual-problem} and the last inequality follows from the optimal solution satisfying $\tilde \mu = \max_{A \in A}\ \sum_{\ell \in A} \tilde \lambda_\ell \leq A_\max$.
    
    Now, we use our termination criterion to ensure that Frank-Wolfe gap is smaller than $\epsilon$, i.e.,
    \begin{align*}
        \max_{\tilde \lambda, \tilde \mu \geq 0}\ \tilde \lambda^\top q - \tilde \mu \cdot B \quad \quad \text{s.t.}\quad \left\{\sum_{\ell \in A} \tilde \lambda_\ell \leq \tilde \mu \quad \forall\ A \in \A\,;\ \lambda\leq a \right\}\quad \leq \quad \lambda^\top q\ -\ \mu \cdot B\ +\ \epsilon
    \end{align*}
    Thus, we get
    \begin{align*}
        \F(q)\quad \leq \quad 2 \gamma  \log_2(4L) \cdot A_\max\ +\ \epsilon\ +\ \lambda^\top q - \mu \cdot B\,.
    \end{align*}
    Combining everything yields the following upper bound on the cost of the algorithm:
    \begin{align}
        \alg\ \leq\ \epsilon\ +\ A_\max \cdot \left(\frac{\sqrt{L}}{2} + 2 \cdot \gamma \cdot \log_2(4L)\right)\ +\ \sum_{\ell \in \L} s_\ell(q_\ell)\ +\ \lambda^\top q - \mu \cdot B \, \tag{$\spadesuit$}.
    \end{align}

    \textbf{Step 2 (Lower bound on optimal Cost).} Recall that
    \begin{align*}
        \opt\ =\ \min_{P_\ell \in \pp_\ell} \quad \sum_{\ell \in \L}\ S_\ell(P_\ell)\ +\ \E_{X_i \sim Q_i(P_i)}[C(X)]\,.
    \end{align*}

    Note that $C(X)$ is an LP and weak duality holds:
    \begin{align*}
        C(X)\ &\geq\ \min_{z,x \geq 0}\ (a - \lambda)^\top z + \lambda^\top X - \mu \cdot B + \sum_{A \in A} \left(\mu - \sum_{\ell \in A} \lambda_\ell \right) \cdot x_A\\
        &\geq \ \lambda^\top X - \mu\cdot B\,,
    \end{align*}
    where the last inequality follows from the feasibility of the solution $(\lambda, \mu)$ with which Algorithm~\ref{alg:dfw} terminates, i.e., $\lambda \leq a$ and $\sum_{\ell \in A} \lambda_\ell \leq \mu$. Therefore, we get
    \begin{align}
        \opt\ &\geq\ \min_{P_\ell \in \pp_\ell} \quad \sum_{\ell \in \L}\ S_\ell(P_\ell)\ +\ \E_{X_\ell \sim Q_\ell(P_\ell)}[\lambda^\top  X - \mu \cdot B] \notag \\
        &=\ \sum_{\ell \in \L}\ \min_{P_\ell \in \pp_\ell}\ S_\ell(P_\ell) + Q_\ell(P_\ell) \cdot \lambda_\ell - \mu \cdot B \notag\\
        &=\ \sum_{\ell \in \L} s_\ell(\lambda_\ell) + q_\ell \cdot \lambda_\ell - \mu \cdot B \tag{$\clubsuit$}\,,
    \end{align}
    where the last equality follows from the definition of $s_\ell(\lambda_\ell)$ and $q_\ell$.

    \textbf{Step 3 (Putting everything together).} Finally, we can combine $(\spadesuit)$ and $(\clubsuit)$ to get the desired bound:
    \begin{align*}
        \alg\ \leq\ \epsilon\ +\ A_\max \cdot \left(\frac{\sqrt{L}}{2} + 2 \cdot \gamma \cdot \log_2(4L)\right)\ +\ \opt\,.  \tag*{\qedhere}
    \end{align*}
\end{proof}\begin{proof}[Proof of \Cref{thm:performance-guarantee}]
    
    Let $(\lambda, \mu)$ be the iterates with which Algorithm~\ref{alg:dfw} terminates. The proof is split into multiple steps: we first upper bound the cost of $\alg$, then lower bound the cost of $\opt$, and finally put together the two bounds. To simplify notation, we set $q_\ell = q_\ell(\lambda_\ell)$ for all $\ell \in \L$.

    \textbf{Step 1 (Upper bound on algorithm's performance). } Applying Lemma~\ref{lemma:fluid-error} yields
    \begin{align*}
        \alg\ \leq\ \sum_{\ell \in \L} s_\ell(\lambda_\ell)\ +\ \F(q)\ +\ A_\max \cdot \frac{\sqrt{L}}{2}\,.
    \end{align*}
    
    Next, we show that approximate strong duality holds and bound $\F(q)$ using the dual optimization problem derived in Lemma~\ref{lemma:dual-problem}. To do so, we leverage Lemma~\ref{lemma:approx-convexity} to get
    \begin{align*}
        \F(q) \quad \leq \quad \min_{p,d} \quad &a^\top p\\
    \text{s.t.} \quad &p + d \geq q\\
    & g(d) \leq B - 2 \gamma \log_2(4L)\\
    &d,\ p\ \in \R_+^L 
    \end{align*}
    where we have used the fact that $\B(d) \leq g(d) + 4 \cdot \gamma$ for all $d \in \R_+^L$. As $g$ is convex, the RHS optimization problem satisfies strong duality due to Slater's condition (pick a $d$ small enough and $p$ small enough to get an strictly interior point). Therefore, we get
    \begin{align*}
       \F(q) \quad &\leq \quad \max_{\tilde \lambda, \tilde \mu \geq 0}\ \min_{p,d \geq 0}\ (a - \tilde\lambda)^\top p + \left(\tilde  \mu \cdot g(d) - \tilde \lambda^\top d \right) - \tilde \mu \cdot (B - 2  \gamma \log_2(4L)) + \tilde \lambda^\top q\\
       &\leq \quad \max_{\tilde \lambda, \tilde \mu \geq 0}\  \min_{p,d \geq 0}\ \ (a - \tilde\lambda)^\top p +\ \left(\tilde  \mu \cdot \B(d) - \tilde \lambda^\top d \right) - \tilde \mu \cdot (B - 2 \gamma \log_2(4L)) + \tilde \lambda^\top q\\
       &\leq \quad \max_{\tilde \lambda, \tilde \mu \geq 0}\ \tilde \lambda^\top q - \tilde \mu \cdot (B - 2 \gamma \log_2(4L)) \quad \quad \text{s.t.}\quad \left\{\sum_{\ell \in A} \tilde \lambda_\ell \leq \tilde \mu\quad \forall\ A \in \A \,;\ \tilde \lambda\leq a \right\}\\
       &\leq \quad 2 \gamma \log_2(4L) \cdot A_\max\ +\  \max_{\tilde \lambda, \tilde \mu \geq 0}\ \tilde \lambda^\top q - \tilde \mu \cdot B \quad \quad \text{s.t.}\quad \left\{\sum_{\ell \in A} \tilde \lambda_\ell \leq \tilde \mu \quad \forall\ A \in \A\,;\ \tilde \lambda\leq a \right\}\,,
    \end{align*}
    where the second-last inequality follows from Lemma~\ref{lemma:dual-problem} and the last inequality follows from the optimal solution satisfying $\tilde \mu = \max_{A \in A}\ \sum_{\ell \in A} \tilde \lambda_\ell \leq A_\max$.
    
    Now, we use our termination criterion to ensure that Frank-Wolfe gap is smaller than $\epsilon$, i.e.,
    \begin{align*}
        \max_{\tilde \lambda, \tilde \mu \geq 0}\ \tilde \lambda^\top q - \tilde \mu \cdot B \quad \quad \text{s.t.}\quad \left\{\sum_{\ell \in A} \tilde \lambda_\ell \leq \tilde \mu \quad \forall\ A \in \A\,;\ \lambda\leq a \right\}\quad \leq \quad \lambda^\top q\ -\ \mu \cdot B\ +\ \epsilon
    \end{align*}
    Thus, we get
    \begin{align*}
        \F(q)\quad \leq \quad 2 \gamma  \log_2(4L) \cdot A_\max\ +\ \epsilon\ +\ \lambda^\top q - \mu \cdot B\,.
    \end{align*}
    Combining everything yields the following upper bound on the cost of the algorithm:
    \begin{align}
        \alg\ \leq\ \epsilon\ +\ A_\max \cdot \left(\frac{\sqrt{L}}{2} + 2 \cdot \gamma \cdot \log_2(4L)\right)\ +\ \sum_{\ell \in \L} s_\ell(q_\ell)\ +\ \lambda^\top q - \mu \cdot B \, \tag{$\spadesuit$}.
    \end{align}

    \textbf{Step 2 (Lower bound on optimal Cost).} Recall that
    \begin{align*}
        \opt\ =\ \min_{P_\ell \in \pp_\ell} \quad \sum_{\ell \in \L}\ S_\ell(P_\ell)\ +\ \E_{X_i \sim Q_i(P_i)}[C(X)]\,.
    \end{align*}

    Note that $C(X)$ is an LP and weak duality holds:
    \begin{align*}
        C(X)\ &\geq\ \min_{z,x \geq 0}\ (a - \lambda)^\top z + \lambda^\top X - \mu \cdot B + \sum_{A \in A} \left(\mu - \sum_{\ell \in A} \lambda_\ell \right) \cdot x_A\\
        &\geq \ \lambda^\top X - \mu\cdot B\,,
    \end{align*}
    where the last inequality follows from the feasibility of the solution $(\lambda, \mu)$ with which Algorithm~\ref{alg:dfw} terminates, i.e., $\lambda \leq a$ and $\sum_{\ell \in A} \lambda_\ell \leq \mu$. Therefore, we get
    \begin{align}
        \opt\ &\geq\ \min_{P_\ell \in \pp_\ell} \quad \sum_{\ell \in \L}\ S_\ell(P_\ell)\ +\ \E_{X_\ell \sim Q_\ell(P_\ell)}[\lambda^\top  X - \mu \cdot B] \notag \\
        &=\ \sum_{\ell \in \L}\ \min_{P_\ell \in \pp_\ell}\ S_\ell(P_\ell) + Q_\ell(P_\ell) \cdot \lambda_\ell - \mu \cdot B \notag\\
        &=\ \sum_{\ell \in \L} s_\ell(\lambda_\ell) + q_\ell \cdot \lambda_\ell - \mu \cdot B \tag{$\clubsuit$}\,,
    \end{align}
    where the last equality follows from the definition of $s_\ell(\lambda_\ell)$ and $q_\ell$.

    \textbf{Step 3 (Putting everything together).} Finally, we can combine $(\spadesuit)$ and $(\clubsuit)$ to get the desired bound:
    \begin{align*}
        \alg\ \leq\ \epsilon\ +\ A_\max \cdot \left(\frac{\sqrt{L}}{2} + 2 \cdot \gamma \cdot \log_2(4L)\right)\ +\ \opt\,.  \tag*{\qedhere}
    \end{align*}
\end{proof}

\subsection{Proof of \Cref{prop:contract_utilization}}

\begin{proof}[Proof of \Cref{prop:contract_utilization}]
    Let $(\lambda, \mu)$ be the iterates returned by \Cref{alg:dfw}, and set $q_\ell \coloneqq q_\ell(\lambda_\ell)$.
    Note that $(0,0)$ is a feasible solution to the dual LP in \Cref{alg:dfw}. Therefore, the termination condition of the \textbf{While} loop implies
    \begin{align*}
        \lambda^\top q - \mu \cdot B \geq - \epsilon\,.
    \end{align*}
    Moreover, since $(\lambda, \mu)$ satisfies $\sum_{\ell \in A} \lambda_\ell \leq \mu$ for all $A \in \A$, we get
    \begin{align*}
        \lambda^\top q = \E_{X \sim q}[\lambda^\top X] \leq \E_{X\sim q}[\mu \cdot \B(X)] = \mu \cdot \B(q)\,.
    \end{align*}

    If $\mu = 0$, then dual feasibility implies $\lambda = 0$, and hence $q(\lambda)=q(0)=\mathbf 1$. In particular, $\B(q(\lambda)) \ge B$, and the claim holds.
    We therefore assume $\mu > 0$ below. Combining the two inequalities yields
    \begin{align}\label{eq:inter-contract-use}
        \B(q) \geq B - \frac{\epsilon}{\mu}\,.
    \end{align}

    Next, note that downward-closedness of $\A$ implies $\{\ell\} \in \A$ for all $\ell$, and therefore $\lambda_\ell \leq \mu$ for all $\ell \in \L$. Using \Cref{assum:lipschitz} (with $\beta = \max_\ell \beta_\ell$) and $q(0)=1$, we have $q_\ell(\lambda_\ell) \ge 1 - \beta_\ell \lambda_\ell$, and hence
    \begin{align*}
        \B(q) \geq \frac{\sum_{\ell \in \L} q_\ell}{u}
        \geq \frac{\sum_{\ell \in \L} (1 - \beta_\ell \lambda_\ell)}{u}
        \geq \frac{L \cdot (1 - \beta \cdot \mu)}{u}
        \geq B \cdot \frac{1 - \beta \cdot \mu}{1 - \nu}\,,
    \end{align*}
    where the first inequality uses $\B(X) \ge |X|/u$ for all $X$, and the last inequality uses $uB < (1-\nu)L$.

    If $\mu \le \nu/\beta$, then $(1-\beta\mu)/(1-\nu) \ge 1$ and hence $\B(q)\ge B$, so the claim holds.
    Otherwise $\mu > \nu/\beta$, and plugging this into \eqref{eq:inter-contract-use} yields
    $\B(q) \ge B - \epsilon/\mu \ge B - \epsilon \cdot (\beta/\nu)$,
    establishing the proposition.
\end{proof}

\subsection{Proof of \Cref{cor:contract_utilization_hp}}

\begin{proof}[Proof of \Cref{cor:contract_utilization_hp}]
    For $Z \in \{0,1\}^L$, changing a single coordinate of $Z$ (adding or removing one load) changes $\B(Z)$ by at most $1$ because $\A$ is downward closed and thus singleton sets are feasible. Therefore, $\B(\cdot)$ satisfies bounded differences with constants $1$, and McDiarmid's inequality gives
    \[
        \B(X) \ge \B(q) - \sqrt{\frac{L \ln(1/\delta)}{2}}
    \]
    with probability at least $1-\delta$.
    Combining with \Cref{prop:contract_utilization} yields
    \[
        \B(X) \ge B - \epsilon \cdot (\beta/\nu) - \sqrt{\frac{L \ln(1/\delta)}{2}}
    \]
    with probability at least $1-\delta$.
    Finally, since $U(X)=\min\{\B(X),B\}$, we have $U(X)\ge B - \epsilon(\beta/\nu) - \sqrt{\frac{L \ln(1/\delta)}{2}}$ on the same event.
\end{proof}

\section{Missing Proofs from \Cref{sec:numerics}}

\begin{proof}[Proof of \Cref{prop:gamma-interval-graph}]
For each color $k\in[K]$, define the per-color contract usage function
\[
\B_k(X)\ \coloneqq\ \max_{t\in\{1,\dots,T\}}\ \sum_{\ell:\,c(\ell)=k,\ t\in I_\ell} X_\ell,
\qquad\text{so that}\qquad
\B(X)=\sum_{k=1}^K \B_k(X).
\]
Fix $X\in\mathbb Z_+^{|\L|}$ and consider $d=X/2$. By the rounding definition,
let
\[
Y \ \coloneqq\ \lfloor X/2\rfloor + Z,
\]
where the $Z_\ell$ are independent and satisfy $Z_\ell\sim\mathrm{Bernoulli}(1/2)$ if $X_\ell$ is odd
(and $Z_\ell\equiv 0$ if $X_\ell$ is even). Then
\[
\B(X/2)=\mathbb E[\B(Y)] = \sum_{k=1}^K \mathbb E[\B_k(Y)].
\]

Write, for each load $\ell$,
\[
Y_\ell \ =\ \frac{X_\ell}{2}+\xi_\ell,
\]
where $\xi_\ell=0$ if $X_\ell$ is even, and if $X_\ell$ is odd then
\[
\xi_\ell = Z_\ell-\frac12 \in \left\{-\frac12,+\frac12\right\},
\qquad \mathbb E[\xi_\ell]=0,
\]
with the $\{\xi_\ell\}$ independent.

Fix a color $k$ and define for each time $t$,
\[
S_{k,t}\ \coloneqq\ \sum_{\ell:\,c(\ell)=k,\ t\in I_\ell} Y_\ell
\;=\;
\underbrace{\frac12\sum_{\ell:\,c(\ell)=k,\ t\in I_\ell} X_\ell}_{\mu_{k,t}}
\;+\;
\underbrace{\sum_{\ell:\,c(\ell)=k,\ t\in I_\ell} \xi_\ell}_{D_{k,t}}.
\]
Then
\[
\mathbb E[\B_k(Y)]
=\mathbb E\big[\max_t S_{k,t}\big]
\le \max_t \mu_{k,t} + \mathbb E\big[\max_t D_{k,t}\big]
= \frac{\B_k(X)}{2} + \mathbb E\big[\max_t D_{k,t}\big].
\]
Summing over $k$ yields
\[
\B(X/2)\ \le\ \frac{\B(X)}{2} \;+\; \sum_{k=1}^K \mathbb E\big[\max_t D_{k,t}\big].
\]

It remains to bound $\mathbb E[\max_{t} D_{k,t}]$. For color $k \in [K]$, define
\[
M_k(X)\ \coloneqq\ \max_{t\in\{1,\dots,T\}}\ \sum_{\ell:\,c(\ell)=k,\ t\in I_\ell} \mathbf 1\{X_\ell\ \text{is odd}\}.
\]
For fixed $t$, the deviation term
\[
D_{k,t}\ =\ \sum_{\ell:\,c(\ell)=k,\ t\in I_\ell}\xi_\ell
\]
is a sum of at most $M_k(X)$
independent, mean-zero random variables supported on $\{-1/2,+1/2\}$.
Hence, $D_{k,t}$ is sub-Gaussian with variance proxy $\sigma_{k,t}^2 \leq M_k(X)/4$
(by Hoeffding's lemma). Therefore, by the standard bound on the expected maximum of $T$
mean-zero $\sigma$-sub-Gaussian random variables (see, e.g., \citealt{vershynin2020high},
Sec.~2.7.3),
\[
\mathbb E\!\left[\max_{t\in\{1,\dots,T\}} D_{k,t}\right]
\ \le\ \sqrt{2\log T}\, \cdot \max_t \sigma_{k,t}
\ \le\ \sqrt{2\log T}\cdot \sqrt{\frac{M_k(X)}{4}}
\ =\ \sqrt{\frac{M_k(X)\log T}{2}}.
\]
Consequently,
\[
\B(X/2)\ \le\ \frac{\B(X)}{2}\;+\;\sqrt{\frac{\log T}{2}}\sum_{k=1}^K \sqrt{M_k(X)}.
\]

Finally, let $L_k\coloneqq |\{\ell\in\L:\ c(\ell)=k\}|$ be the number of loads of color $k$.
Since each term $(X_\ell\bmod 2)\le 1$, for every $t$,
\[
\sum_{\ell:\,c(\ell)=k,\ t\in I_\ell} (X_\ell\bmod 2)\ \le\ \sum_{\ell:\,c(\ell)=k} 1 \ =\ L_k,
\]
hence $M_k(X)\le L_k$. Using Cauchy--Schwarz and $\sum_{k=1}^K L_k = |\L| = L$,
\[
\sum_{k=1}^K \sqrt{M_k(X)}
\le \sum_{k=1}^K \sqrt{L_k}
\le \sqrt{K\sum_{k=1}^K L_k}
= \sqrt{KL}.
\]
Plugging in,
\[
\B(X/2)\ \le\ \frac{\B(X)}{2}\;+\;\sqrt{\frac{KL\log T}{2}}
\ \le\ \frac{\B(X)}{2}\;+\;\sqrt{KL\log T}.
\]
Thus, \Cref{assum:contract-subs} is satisfied with $\gamma=\sqrt{L\cdot K\log(T)/2}$.
\end{proof}

\end{document}